%% file: bisection.tex
\documentclass[11pt]{article}
\usepackage[margin=1in]{geometry}
\usepackage[latin2]{inputenc}
\usepackage[english]{babel}
\usepackage[cmex10]{amsmath}
\usepackage{todonotes}
\usepackage{amsthm}
\usepackage{amssymb}
\usepackage{xspace}
\usepackage{comment}
\usepackage{letltxmacro}
\usepackage{comment}
\usepackage{enumerate}
\usepackage{tikz}

\usepackage{graphicx}
\usepackage{caption}
\usepackage{subcaption}

\newtheorem{theorem}{Theorem}[section]
\newtheorem{lemma}[theorem]{Lemma}
\newtheorem{claim}[theorem]{Claim}
\newtheorem{corollary}[theorem]{Corollary}

\theoremstyle{definition}
\newtheorem{definition}[theorem]{Definition}

\newcommand{\hedge}{F} 
\newcommand{\hpname}{{{\sc Hypergraph Painting}}\xspace}
\newcommand{\hp}{{{\sc HP}}\xspace}

\newcommand{\col}{\mathrm{col}}
\newcommand{\parent}{\mathrm{parent}}

\newcommand{\Oh}{\ensuremath{\mathcal{O}}}
\newcommand{\Ohstar}{\ensuremath{\Oh^\star}}
\newcommand{\R}{\ensuremath{\mathbb{R}}}

\def\cqedsymbol{\ifmmode$\lrcorner$\else{\unskip\nobreak\hfil
\penalty50\hskip1em\null\nobreak\hfil$\lrcorner$
\parfillskip=0pt\finalhyphendemerits=0\endgraf}\fi} 

\newcommand{\cqed}{\renewcommand{\qed}{\cqedsymbol}}

\newcommand{\executeiffilenewer}[3]{%
\ifnum\pdfstrcmp{\pdffilemoddate{#1}}%
{\pdffilemoddate{#2}}>0%
{\immediate\write18{#3}}\fi%
} 
\newcommand{\defproblemG}[3]{
  \vspace{2mm}
\noindent\fbox{
  \begin{minipage}{0.96\textwidth}
  #1 \\
  {\bf{Input:}} #2  \\
  {\bf{Goal:}} #3
  \end{minipage}
  }
  \vspace{2mm}
}

\newcommand{\B}{{\bf B}}
\newcommand{\W}{{\bf W}}
\newcommand{\C}{\ensuremath{\mathcal{C}}\xspace}
\newcommand{\touch}{\sim}

\title{Minimum Bisection is Fixed Parameter Tractable%
\thanks{M. Cygan is supported by the Polish National Science Centre, grant n. UMO-2013/09/B/ST6/03136.
D. Lokshtanov is supported by the BeHard grant under the recruitment programme of the of Bergen Research Foundation.
The research of M. Pilipczuk and M. Pilipczuk leading to these results has received funding from the European Research Council under the European Union's Seventh Framework Programme (FP/2007-2013) / ERC Grant Agreement n. 267959.
S. Saurabh is supported by PARAPPROX, ERC starting grant no. 306992.}}
\author{
  Marek Cygan
  \thanks{
    Institute of Informatics, University of Warsaw, Poland,
      \texttt{cygan@mimuw.edu.pl}.
  }
  \and
  Daniel Lokshtanov
  \thanks{
    Department of Informatics, University of Bergen, Norway, \texttt{daniello@ii.uib.no}.
  }
  \and
  Marcin Pilipczuk
  \thanks{
    Department of Informatics, University of Bergen, Norway, \texttt{Marcin.Pilipczuk@ii.uib.no}.
  }
  \and
  Micha\l{} Pilipczuk
  \thanks{
    Department of Informatics, University of Bergen, Norway, \texttt{michal.pilipczuk@ii.uib.no}.
  }
  \and 
  Saket Saurabh\thanks{
    Institute of Mathematical Sciences, India \texttt{saket@imsc.res.in} 
    Department of Informatics, University of Bergen, Norway, \texttt{Saket.Saurabh@ii.uib.no}.
  }
  }

\date{}

\begin{document}

\maketitle

\begin{abstract}
\input{abstract}
\end{abstract}

\newpage
\section{Introduction}\label{sec_intro}
\input{intro}

\section{Preliminaries}\label{sec:prelims}
\input{prelims}

\section{Decomposition}\label{sec:decomp}
\input{decomposition}

\section{Bisection}\label{sec:bisection}
\input{dp}

\section{Weighted variant}\label{sec:weights}
\input{weights}

\section{$\alpha$-edge-separators}\label{sec:alpha}
\input{alpha-sep}

\section{Conclusions}\label{sec:conc}
\input{conclusion}

\section*{Acknowledgements}

We would like to thank Rajesh Chitnis, Fedor Fomin, MohammadTaghi Hajiaghayi and M. S. Ramanujan
for earlier discussions on this subject.

We also acknowledge the very inspiring atmosphere of the Dagstuhl seminar 13121, where the authors discussed the core ideas leading to this work.

\bibliographystyle{abbrv}
\bibliography{bisection}

\end{document}

%% file: abstract.tex
In the classic {\sc Minimum Bisection} problem we are given as input a graph $G$ and an integer~$k$. The task is to determine whether there is a partition of $V(G)$ into two parts $A$ and $B$ such that $||A|-|B|| \leq 1$ and there are at most $k$ edges with one endpoint in $A$ and the other in $B$. In this paper we give an algorithm for {\sc Minimum Bisection} with running time
$\Oh(2^{\Oh(k^{3})}n^3 \log^3 n)$. This is the first fixed parameter tractable algorithm for {\sc Minimum Bisection}. 
%
At the core of our algorithm lies a new decomposition theorem that states that every graph $G$ can be decomposed by small separators into parts where each part is ``highly connected''
in the following sense: any cut of bounded size can separate only a limited number of vertices from each part of the decomposition. Our techniques generalize to the weighted setting, where  we seek for a bisection of minimum weight among solutions that contain at most $k$ edges.

%% file: intro.tex
In the {\sc Minimum Bisection} problem the input is a graph $G$ on $n$ vertices together with an integer $k$, and the objective is to find a partition of the vertex set into two parts $A$ and $B$ such that $|A| = \lfloor \frac{n}{2} \rfloor$, $|B| = \lceil \frac{n}{2} \rceil$, and there are at most $k$ edges with one endpoint in $A$ and the other endpoint in $B$. The problem can be seen as a variant of {\sc Minimum Cut}, and is one of the classic NP-complete problems~\cite{garey1979computers}. {\sc Minimum Bisection} has been studied extensively from the perspective of approximation algorithms~\cite{FeigeKN00,FeigeK02,KhotV05,Racke08}, heuristics~\cite{BuiHJL89,BuiS02} and average case complexity~\cite{BuiCLS87}. 

In this paper we consider the complexity of {\sc Minimum Bisection} when the solution size $k$ is small relative to the input size $n$. A na\"ive brute-force algorithm solves the problem in time $n^{\Oh(k)}$. 
Until this work, it was unknown whether there exists a {\em fixed parameter tractable} algorithm, that is an algorithm with running time $f(k)n^{\Oh(1)}$, for the {\sc Minimum Bisection} problem. In fact  {\sc Minimum Bisection} was one of very few remaining classic NP-hard graph problems whose parameterized complexity status was unresolved. Our main result is the first fixed parameter tractable algorithm for {\sc Minimum Bisection}.
\begin{theorem}\label{thm:bisectionFPT} 
{\sc Minimum Bisection} admits an $\Oh(2^{\Oh(k^3)}n^3 \log^3 n)$ time algorithm. 
\end{theorem}
\noindent
Theorem~\ref{thm:bisectionFPT} implies that {\sc Minimum Bisection} can be solved in polynomial time for $k = \Oh(\sqrt[3]{\log n})$. In fact, our techniques can be generalized to solve the more general problem where the target $|A|$ is given as input, the edges have non-negative weights, and the objective is to find, among all partitions of $V(G)$ into $A$ and $B$ such that $A$ has the prescribed size and there are at most $k$ edges between $A$ and $B$, such a partition where the total weight of the edges between $A$ and $B$ is minimized.

\bigskip
\noindent
{\bf Our methods.} The crucial technical component of our result is a new graph decomposition theorem. Roughly speaking, the theorem states that for any $k$, every graph $G$ may be decomposed in a tree-like fashion by separators of size $2^{\Oh(k)}$ such that each part of the decomposition is ``highly connected''. To properly define what we mean by ``highly connected'' we need a few definitions. A {\em separation} of a graph $G$ is a pair $A,B \subseteq V(G)$ such that $A \cup B = V(G)$ and there are no edges between $A \setminus B$ and $B \setminus A$. The {\em order} of the separation $(A,B)$ is $|A \cap B|$. A vertex set $X \subseteq V(G)$ is called $(q,k)$-{\em unbreakable} if every separation $(A,B)$ of order at most $k$ satisfies $|(A \setminus B) \cap X| \leq q$ or $|(B \setminus A) \cap X| \leq q$. The parts of our decomposition will be ``highly connected'' in the sense that they are $(2^{\Oh(k)},k)$-unbreakable. We can now state the decomposition theorem as follows.
\begin{theorem}\label{thm:main_dec}
There is an algorithm that given $G$ and $k$ runs in time $\Oh(2^{\Oh(k^2)} n^2m)$ and outputs a tree-decomposition $(T,\beta)$ of $G$ such that $(i)$ for each $a \in V(T)$, $\beta(a)$ is $(2^{\Oh(k)},k)$-unbreakable in $G$,  $(ii)$ for each $ab \in E(T)$ we have that $|\beta(a) \cap \beta(b)| \leq 2^{\Oh(k)}$, and  $\beta(a) \cap \beta(b)$ is $(2k,k)$-unbreakable in $G$.
\end{theorem}
Here $\beta(a)$ denotes the bag at node $a \in V(T)$; the completely formal definition of tree-decompositions may be found in the preliminaries.
It is not immediately obvious that a set $X$ which is $(q,k)$-unbreakable is ``highly connected''.  To get some intuition it is helpful to observe that if a set $X$ of size at least $3q$ is $(q,k)$-unbreakable then removing any $k$ vertices from $G$ leaves almost all of $X$, except for at most $q$ vertices, in the same connected component. In other words, one cannot separate two large chunks of $X$ with a small separator. From this perspective Theorem~\ref{thm:main_dec} can be seen as an approximate way to ``decompose a graph by $k$ vertex-cuts into it's $k+1$-connected components''~\cite{DiestelDecomposition}, which is considered an important quest in structural graph theory. The proof strategy of Theorem~\ref{thm:main_dec} is inspired by the recent decomposition theorem of Marx and Grohe~\cite{marx-grohe} for graphs excluding a topological subgraph. Contrary to the approach of Marx and Grohe~\cite{marx-grohe}, however, the crucial technical tool we use to decompose the graph are the important separators of Marx~\cite{Marx06}.

Our algorithm for {\sc Minimum Bisection} applies Theorem~\ref{thm:main_dec} and then proceeds by performing bottom up dynamic programming on the tree-decomposition. The states in the dynamic program are similar to the states in the dynamic programming algorithm for {\sc Minimum Bisection} on graphs of bounded treewidth~\cite{JansenKLS05}. Property (ii) of  Theorem~\ref{thm:main_dec} ensures that the size of the dynamic programming table is upper bounded by $2^{\Oh(k^2)}n^{\Oh(1)}$. For graphs of bounded treewidth all bags have small size, making it easy to compute the dynamic programming table at a node $b$ of the decomposition tree, if the tables for the children of $b$ have already been computed. In our setting we do not have any control over the size of the bags, we only know that they are  $(2^{\Oh(k)},k)$-unbreakable. We show that the sole assumption that the bag at $b$ is $(2^{\Oh(k)},k)$-unbreakable is already sufficient to efficiently compute the table at $b$ from the tables of its children, despite no a priori guarantee on the bag's size. The essence of this step is an application of the ``randomized contractions'' technique~\cite{rand-contractions}.

We remark here that the last property of the decomposition of Theorem~\ref{thm:main_dec}
--- the one that asserts that adhesions $\beta(a) \cap \beta(b)$ are $(2k,k)$-unbreakable in $G$
--- is not essential to establish the fixed-parameter tractability of \textsc{Minimum Bisection}.
This high unbreakability of adhesions is used to further limit
the number of states of the dynamic programming,
decreasing the dependency on $k$ in the algorithm of Theorem~\ref{thm:bisectionFPT}
from double- to single-exponential.

\bigskip
\noindent
{\bf Related work on balanced separations.} There are several interesting results concerning the parameterized complexity of finding balanced separators in graphs. Marx~\cite{Marx06} showed that the the vertex-deletion variant of the bisection problem is W[1]-hard. In {\sc Minimum Vertex Bisection} the task is to partition the vertex set into three parts $A$, $S$ and $B$ such that $|S| \leq k$ and $|A| = |B|$, and there are no edges between $A$ and $B$. It is worth mentioning that the hardness result of Marx~\cite{Marx06} applies to the more general problem where $|A|$ is given as input, however the hardness of {\sc Minimum Vertex Bisection} easily follows from the results presented in~\cite{Marx06}. 

As the vertex-deletion variant of the bisection problem is W[1]-hard,
we should not expect that our approach would work also in this case.
Observe that one can compute the decomposition of Theorem~\ref{thm:main_dec}
and define the states of the dynamic programming over the tree decomposition, as it
is done for graphs of bounded treewidth. However, we are unable 
to perform the computations needed for one bag of the decomposition.
Moreover, it is not only the artifact of the ``randomized contractions'' technique,
but the hard instances obtained from the reduction of~\cite{Marx06}
are in fact highly unbreakable by our definition, and Theorem~\ref{thm:main_dec}
would return a trivial decomposition.

Feige and Mahdian~\cite{FeigeM06} studied cut problems that may be considered as approximation variants of {\sc Minimum Bisection} and  {\sc Minimum Vertex Bisection}.  We say that a vertex (edge) set $S$ is an $\alpha$-(edge)-separator if every connected component of $G \setminus S$ has at most $\alpha n$ vertices. The main result of Feige and Mahdian~\cite{FeigeM06} is a randomized algorithm that given an integer $k$, $\frac{2}{3} \leq \alpha < 1$ and $\epsilon > 0 $ together with a graph $G$ which has an $\alpha$-separator of size at most $k$, outputs in time $2^{f(\epsilon) k}n^{\Oh(1)}$ either an $\alpha$-separator of size at most $k$ or an $(\alpha + \epsilon)$-separator of size strictly less than $k$. They also give a deterministic algorithm with similar running time for the edge variant of this problem. To complement this result they show that, at least for the vertex variant, the exponential running time dependence on $1/\epsilon$ is unavoidable. Specifically, they prove that for any $\alpha > \frac{1}{2}$ finding an $\alpha$-separator of size $k$ is W[1]-hard, and therefore unlikely to admit an algorithm with running time $f(k)n^{\Oh(1)}$, for any function $f$. On the other hand, our methods imply a $2^{\Oh(k^3)}n^{\Oh(1/\alpha)}$ time algorithm for finding an $\alpha$-edge-separator of size at most $k$, for any $\alpha > 0$.
 
{\sc Minimum Bisection} on planar graphs was shown to be fixed parameter tractable by Bui and Peck~\cite{BuiP92}. It is interesting to note that {\sc Minimum Bisection} is not known to be NP-hard on planar graphs, and the complexity of {\sc Minimum Bisection} on planar graphs remains a challenging open problem. More recently, van Bevern et al.~\cite{vanBevern13} used the treewidth reduction technique of Marx et al.~\cite{MarxOR13} to give a fixed parameter tractable algorithm for {\sc Minimum Bisection} for the special case when removing the cut edges leaves a constant number of connected components. Their algorithm also works for the vertex-deletion variant the same restrictions. Since {\sc Minimum Vertex Bisection} 
is known to be W[1]-hard, it looks difficult to extend their methods to give a  fixed parameter tractable algorithm for {\sc Minimum Bisection} without any restrictions.
Thus, Theorem~\ref{thm:bisectionFPT} resolves an open problem of van Bevern et al.~\cite{vanBevern13} on the existence of such an algorithm.

\bigskip
\noindent
{\bf Related work on graph decompositions.}
The starting point of our decomposition theorem is the ``recursive understanding'' technique pioneered by Grohe et al.~\cite{GroheKMW11}, and later used by Kawarabayashi and Thorup~\cite{KawarabayashiT11} and by Chitnis et al.~\cite{rand-contractions} to design a number of interesting parameterized algorithms for cut problems. Recursive understanding can be seen as a reduction from a parameterized problem on general graphs to the same problem on graphs with a particular structure. Grohe et al.~\cite{GroheKMW11} essentially use recursive understanding to reduce the problem of deciding whether $G$ contains $H$ as a topological subgraph to the case where $G$ either excludes a clique on $f(|H|)$ vertices as a {\em minor} or contains at most $f(|H|)$ vertices of degree more than $f(|H|)$, for some function $f$. Marx and Grohe~\cite{marx-grohe} subsequently showed that any graph which excludes $H$ as a topological subgraph can be decomposed by small separators, in a tree-like fashion, into parts such that each part either excludes a clique on $f(|H|)$ vertices as a minor or contains at most $f(|H|)$ vertices of degree more than $f(|H|)$, for some function $f$. Thus, the decomposition theorem of Marx and Grohe~\cite{marx-grohe} can be seen as a ``structural'' analogue of the recursive understanding technique for topological subgraph containment.

Both Kawarabayashi and Thorup~\cite{KawarabayashiT11} and Chitnis et al.~\cite{rand-contractions} apply recursive understanding to reduce certain parameterized cut problems on general graphs to essentially the same problem on a graph $G$ where $V(G)$ is $(f(k),k)$-unbreakable for some function $f$. Then they proceed to show that the considered problem becomes fixed parameter tractable on $(f(k),k)$-unbreakable graphs. Observe that {\sc Minimum Bisection} on  $(f(k),k)$-unbreakable graphs is trivially fixed parameter tractable - if the number of vertices is more than $2f(k)$ we can immediately say no, while if the number of vertices is at most $2f(k)$, then a brute force algorithm is already fixed parameter tractable. More importantly, it turns out that even the more general problem where $|A|$ is given on the input can be solved in fixed parameter tractable time on $(f(k),k)$-unbreakable graphs via an application of the ``randomized contractions'' technique of Chitnis et al~\cite{rand-contractions}. It is therefore very natural to try to use recursive understanding in order to reduce {\sc Minimum Bisection} on general graphs to {\sc Minimum Bisection} on $(f(k),k)$-unbreakable graphs.

Unfortunately, it seems very difficult to pursue this route. In particular, recursive understanding works by cutting the graph into two parts by a small separator, ``understanding'' the easier of the two parts recursively, and then replacing the``understood'' part by a constant size gadget. For {\sc Minimum Bisection} it seems unlikely that the understood part can be emulated by any constant size gadget because of the balance constraint in the problem definition. Intuitively, we would need to encode the behaviour of the understood part for every possible cardinality of $A$, which gives us amount of information that is not bounded by a function of $k$. The issue has strong connections to the fact that the best known algorithm for {\sc Minimum Bisection} on graphs of bounded treewidth is at least quadratic~\cite{JansenKLS05} rather than linear. 

At this point our decomposition theorem comes to the rescue. It precisely allows us to {\em structurally} decompose the graph in a tree-like fashion into $(f(k),k)$-unbreakable parts, which provides much more robust foundations for further algorithmic applications. Thus, essentially our decomposition theorem does the same for cut problems as the decomposition theorem of Marx and Grohe~\cite{marx-grohe}  does for topological subgraph containment. Notably, the ``recursive understanding'' step used by Kawarabayashi and Thorup~\cite{KawarabayashiT11} and Chitnis et al.~\cite{rand-contractions} for their problems could be replaced by dynamic programming over the tree-decomposition given by Theorem~\ref{thm:main_dec}.

We remark here that it has been essentially known, and observed earlier by Chitnis, Cygan and Hajiaghayi (private communication), that {\sc Minimum Bisection} can be solved in FPT time on sufficiently
unbreakable graphs via the ``randomized contractions'' technique. Furthermore, although our application of this framework to handle one bag of the decomposition
is more technical than in~\cite{rand-contractions}, due to the presence
of the information for children bags, it uses no novel tools compared to~\cite{rand-contractions}.
Hence, we emphasize that our main technical contribution is the decomposition theorem (Theorem~\ref{thm:main_dec}), with the fixed-parameter algorithm for {\sc Minimum Bisection}
being its corollary via an involved application of known techniques.

\bigskip
\noindent
\textbf{Organisation of the paper.}
After setting up notation and recalling useful results on (important) separators
in Section~\ref{sec:prelims}, we turn our attention to the decomposition theorem and prove
Theorem~\ref{thm:main_dec} in Section~\ref{sec:decomp}.
The algorithm for \textsc{Minimum Bisection} in the unweighted setting,
promised by Theorem~\ref{thm:bisectionFPT}, is presented in Section~\ref{sec:bisection}.
We discuss the weighted extension in Section~\ref{sec:weights} and how to find
an $\alpha$-edge-separator of size at most $k$ in Section~\ref{sec:alpha}.
Section~\ref{sec:conc} concludes the paper.

%% file: prelims.tex
We use standard graph notation, see e.g.~\cite{diestel}. We use $n$ and $m$ to denote cardinalities of the vertex and edge sets, respectively, of a given graph provided it is clear from the context.
We begin with some definitions and known results on separators and separations in graphs.

\begin{definition}[\bf separator]
For two sets $X,Y \subseteq V(G)$ a set $W \subseteq V(G)$
is called an $X-Y$ separator if in $G \setminus W$ no connected component contains
a vertex of $X$ and a vertex of $Y$ at the same time.
\end{definition}

\begin{definition}[\bf separation]
A pair $(A,B)$ where $A \cup B = V(G)$ is called a {\em separation}
if $E(A \setminus B, B \setminus A) = \emptyset$.
The {\em order} of a separation $(A,B)$ is defined as $|A \cap B|$.
\end{definition}

\begin{definition}[\bf important separator]
An inclusion-wise minimal $X-Y$ separator $W$ is called
an {\emph{important}} $X-Y$ separator if there is no
$X-Y$ separator $W'$ with $|W'| \le |W|$ and
$R_{G\setminus W}(X \setminus W) \subsetneq R_{G \setminus W'}(X \setminus W')$,
where $R_H(A)$ is the set of vertices reachable from $A$ in the graph $H$.
\end{definition}

\begin{lemma}[\cite{imp-seps-chen,multicut}]
\label{lem:imp-seps}
For any two sets $S,T \subseteq V(G)$ there are at most $4^k$
important $S-T$ separators of size at most $k$
and one can list all of them in $\Oh(4^k k (n+m))$ time.
\end{lemma}

We proceed to define tree-decompositions. For a rooted tree $T$ and a non-root node $t \in V(T)$, by $\parent(t)$
we denote the parent of $t$ in the tree $T$. For two nodes $u,t\in T$, we say that $u$ is a {\emph{descendant}} of $t$, denoted $u\preceq t$, if $t$ lies on the unique path connecting $u$ to the root. Note that every node is thus its own descendant.

\begin{definition}[\bf tree decomposition]
A {\em tree decomposition} of a graph $G$ is a pair $(T,\beta)$, where $T$ 
is a rooted tree and $\beta : V(T) \to 2^{V(G)}$ is a mapping such that:
\begin{itemize}
  \item for each node $v \in V(G)$ the set $\{t \in V(G) | v \in \beta(t)\}$ induces a nonempty and connected subtree of~$T$,
  \item for each edge $e \in E(G)$ there exists $t \in V(T)$ such that $e \subseteq \beta(t)$.
\end{itemize}
\end{definition}

The set $\beta(t)$ is called the {\emph{bag at $t$}}, while sets $\beta(u)\cap \beta(v)$ for $uv\in E(T)$ are called {\emph{adhesions}}. Following the notation from~\cite{marx-grohe}, for a tree decomposition $(T,\beta)$ of a graph $G$ we define auxiliary mappings $\sigma,\gamma:V(T) \to 2^{V(G)}$ as 
\begin{align*}
\sigma(t) & = \begin{cases} \emptyset & \text{if t is the root of }T \\ \beta(t) \cap \beta(\parent(t)) & \text{otherwise}\end{cases} \\
\gamma(t) & = \bigcup_{u\preceq t} \beta(u)
\end{align*}

Finally, we proceed to the definition of unbreakability.

\begin{definition}[\bf $(q,k)$-unbreakable set]
We say that a set $A$ is $(q,k)$-{\emph{unbreakable}},
if for any separation $(X,Y)$ of order at most $k$
we have $|(X \setminus Y) \cap A| \le q$ or $|(Y \setminus X) \cap A| \le q$.
Otherwise $A$ is $(q,k)$-{\emph{breakable}}, and any separation $(X,Y)$ certifying this is called a {\emph{witnessing separation}}.
\end{definition}

Let us repeat the intuition on unbreakable sets from the introduction.
If a set $X$ of size at least $3q$ is $(q,k)$-unbreakable then removing any $k$ vertices from $G$
leaves almost all of $X$, except for at most $q$ vertices, in the same connected component.
In other words, one cannot separate two large chunks of $X$ with a small separator.

Observe that if a set $A$ is $(q,k)$-unbreakable in $G$, then any of its subset
$A' \subseteq A$ is also $(q,k)$-unbreakable in $G$. 
Moreover, if $A$ is $(q,k)$-unbreakable in $G$, then $A$ is also $(q,k)$-unbreakable in any supergraph of $G$. For a small set $A$ it is easy to efficiently verify whether $A$ is $(q,k)$-unbreakable in $G$, or to find a witnessing separation.

\begin{lemma}
\label{lem:breaking}
Given a graph $G$, a set $A \subseteq V(G)$ and an integer $q$
one can check in $\Oh(|A|^{2q+2} k(n+m))$ time whether $A$ is $(q,k)$-unbreakable in $G$,
and if not, then find a separation $(X,Y)$ of order at most $k$
such that $|(X\setminus Y) \cap A| > q$ and $|(Y \setminus X) \cap A| > q$.
\end{lemma}

\begin{proof}
Our algorithm guesses, by trying all possibilities, 
two disjoint subsets $X_0,Y_0 \subseteq A$ of $q+1$ vertices each.
Having fixed $X_0$ and $Y_0$ we may, in $\Oh(k(n+m))$ time by applying $(k+1)$ rounds of Ford-Fulkerson algorithm,
find a minimum $X_0-Y_0$ separator in $G$, or conclude that its size is larger than $k$.
If a separator $Z$ of size at most $k$ exists, then obtain a separation $(X',Y')$ as follows:
set $X' \cap Y' = Z$, add connected components of $G \setminus Z$
intersecting $X_0$ to $X' \setminus Y'$, add connected
components intersecting $Y_0$ to $Y' \setminus X'$, and distribute
all the other connected component arbitrarily between $X'\setminus Y'$ and $Y'\setminus X'$.
Observe that since $|X'\cap Y'|\leq k$, $|X'\setminus Y'|\geq |X_0|=q+1$, and $|Y'\setminus X'|\geq |Y_0|=q+1$, then separation $(X',Y')$ witnesses that $A$ is $(q,k)$-breakable, and thus can be output by the algorithm. If for none of the pairs $(X_0,Y_0)$ admits a $X_0-Y_0$ separator of size at most $k$, then we conclude that $A$ is $(q,k)$-unbreakable.

It remains to argue that if $A$ is $(q,k)$-breakable, then for some pair $(X_0,Y_0)$ the minimum $X_0-Y_0$ separator has size at most $k$. Indeed, let $(X,Y)$ be any separation of order at most $k$ witnessing that $A$ is $(q,k)$-breakable, and let $X_0 \subseteq X \setminus Y$ and $Y_0 \subseteq Y \setminus X$ be any subsets of size $q+1$. Then $X\cap Y$ is an $X_0-Y_0$ separator of size at most $k$.
\end{proof}

%% file: decomposition.tex
We now restate our decomposition theorem in a slightly stronger form that will emerge from the proof.

\begin{theorem}
\label{thm:decomposition-main}
There is an $\Oh(2^{\Oh(k^2)} n^2 m)$ time algorithm
that, given a connected graph $G$ together with
an integer $k$, computes a tree decomposition $(T,\beta)$ of $G$ with at most $n$ nodes
such that the following conditions hold:
\begin{enumerate}[(i)]
  \item for each $t \in V(T)$, the graph $G[\gamma(t)]\setminus \sigma(t)$ is connected and $N(\gamma(t)\setminus \sigma(t))=\sigma(t)$ ,
  \item for each $t \in V(T)$, the set $\beta(t)$ is $(2^{\Oh(k)},k)$-unbreakable in $G[\gamma(t)]$,
  \item for each non-root $t \in V(T)$, we have that $|\sigma(t)|\leq 2^{\Oh(k)}$ and $\sigma(t)$ is $(2k,k)$-unbreakable in $G[\gamma(\parent(t))]$.
\end{enumerate}
\end{theorem}

\subsection{Proof overview}

We first give an overview of the proof of Theorem~\ref{thm:decomposition-main},
ignoring the requirement that each adhesion is supposed to be $(2k,k)$-unbreakable.
As discussed in the introduction, this property is only used to improve the running time
of the algorithm, and is not essential to establish fixed-parameter tractability.

We prove the decomposition theorem using a recursive approach, similar to the standard framework used for instance by Robertson and Seymour~\cite{gm13} or by Marx and Grohe~\cite{marx-grohe}.
That is, in the recursive step
we are given a graph $G$ together with a relatively small set $S\subseteq V(G)$ (i.e., of size
bounded by $2^{\Oh(k)}$), and
our goal is to construct a decomposition of $G$ satisfying the requirements
of Theorem~\ref{thm:decomposition-main} with an additional property that $S$
is contained in the root bag of the decomposition.
The intention is that the recursive step is invoked on some subgraph of the input graph,
and the set $S$ is the adhesion towards the decomposition of the rest of the graph.

Henceforth we focus on one recursive step, and consider three cases.
In the base case, if $|S| \leq 3k$, we add an arbitrary vertex to $S$ and repeat.
In what follows, we assume $|S| > 3k$.

First, assume that $S$ is $(2k,k)$-breakable in $G$, and let $(X,Y)$ be the witnessing
separation. We proceed in a standard manner (cf.~\cite{gm13}): we create a root bag
$A := S \cup (X \cap Y)$, for each connected component $C$ of $G \setminus A$
recurse on $(N_G[C], N_G(C))$, and glue the obtained trees as children of the root bag.
It is straightforward from the definition of the witnessing separation that
in every recursive call we have $|N_G(C)| \leq |S|$.
Moreover, clearly $|A| \leq |S|+k$ and hence $A$ is appropriately unbreakable.

In the last, much more interesting case the adhesion $S$ turns out to be $(2k,k)$-unbreakable.
Hence, any separation $(X,Y)$ in $G$ partitions $S$ very unevenly: almost the entire set $S$,
up to $\Oh(k)$ elements, lies on only one side of the separation. 
Let us call this side the ``big'' side, and the second side the ``small'' one.

The main idea now is as follows: if, for each $v \in V(G)$, we mark all important separators
of size $\Oh(k)$ between $v$ and $S$, then the marked vertices
will separate all ``small'' sides
of separations from the set $S$.
Let $B$ be the set of marked vertices and let $A$ be the set of all vertices of $G$ that are either
in $B \cup S$, or are not separated from $S$ by any of the considered important separator.
We observe that the strong structure of important separators ---
in particular, the single-exponential bound on the number of important separators for one vertex $v$
--- allows us to argue that each connected component $C$ of $G \setminus A$ that
is separated by some important separator from $S$ has only bounded number of neighbours in
$A$.
Moreover, the fact that we cut all ``small'' sides of separations implies
that $A$ is appropriately unbreakable in $G$.
Hence, we may recurse, for each connected component $C$ of $G \setminus A$
that is separated by some important separator from $S$, on $(N_G[C], N_G(C))$,
and take $A$ as a root bag.

The section is organised as follows.
In Section~\ref{ssec:chips} we define formally the notion of \emph{chips}, that are parts
of the graph cut out by important separators, 
 and provide all the properties that play crucial role in Section~\ref{ssec:decomposition-local}.
 In Section~\ref{ssec:decomposition-local} we also show how to proceed with the case
 $S$ being unbreakable, that is, how extract the root bag containing $S$ by cutting away all the chips.
 In Section~\ref{ssec:strengthen} we perform some technical augmentation to ensure
 that the adhesions are $(2k,k)$-unbreakable.
 Finally in Section~\ref{ssec:decomposition-main} we combine the obtained results and construct the main decomposition of Theorem~\ref{thm:decomposition-main}.

\subsection{Chips}
\label{ssec:chips}

In this subsection we define fragments 
of the graph which are easy to chip (i.e. cut out of the graph)
from some given set of vertices $S$,
and show their basic properties.

\begin{definition}[\bf chips]
\label{def:chip}
For a fixed set of vertices $S \subseteq V$,
a subset $C \subseteq V$ is called a {\em chip}, if 
\begin{enumerate}[(a)]
  \item \label{def:chip:a} $G[C]$ is connected,
  \item \label{def:chip:b} $|N(C)| \le 3k$,
  \item \label{def:chip:c} $N(C)$ is an important $C-S$ separator.
\end{enumerate}
Let \C be the set of all inclusionwise maximal chips.
\end{definition}

The following lemma is straightforward from the definition of important separators.
\begin{lemma}\label{lem:chips-impseps}
For any nonempty set $C \subseteq V(G)$ such that $G[C]$ is connected, the following conditions are equivalent.
\begin{enumerate}[(i)]
\item $N(C)$ is an important $C-S$ separator;
\item for any $v \in C$, $N(C)$ is an important $v-S$ separator;
\item there exists $v \in C$ such that $N(C)$ is an important $v-S$ separator.
\end{enumerate}
\end{lemma}

%

Note also that for a connected set of vertices $D$ and any important $D-S$ separator $Z$ of size at most $3k$ that is disjoint with $D$, the set of vertices reachable from $D$ in $G\setminus Z$ forms a chip. We now show how to enumerate inclusion-wise maximal chips.
\begin{lemma}
\label{lemma:chips-time-bound}
Given a set $S \subseteq V(G)$ one can compute the set~$\C$
of all inclusion-wise maximal chips in $\Oh(2^{\Oh(k)} n(n+m))$ time. In particular, $|\C|\leq 4^{3k} n$.
\end{lemma}
\begin{proof}
For any $v \in V$, we use Lemma~\ref{lem:imp-seps} to enumerate the set $\mathcal{Z}_v$ of all important $v-S$ separators of size at most $3k$.
Recall that for any $Z \in \mathcal{Z}_v$, the set $R_{G \setminus Z}(v)$ is the vertex set of the connected component
of $G \setminus Z$ containing $v$.
Define $\mathcal{A}_v = \{R_{G \setminus Z}(v) : Z \in \mathcal{Z}_v\}$
and let $\mathcal{C}_v$ be the set of inclusion-wise maximal elements of $\mathcal{A}_v$.
By Lemma~\ref{lem:chips-impseps} we infer that if some chip $C\in \mathcal{A}_v$ is not inclusion-wise maximal, then there exists $C'\in \mathcal{A}_v$ such that $C\subsetneq C'$. Therefore, we have that $\C = \bigcup_{v \in V(G)} \mathcal{C}_v$.

As $|\mathcal{Z}_v| \leq 4^{3k}$ for any $v \in V(G)$, the bound on $|\C|$ follows.
For each $v \in V(G)$, the sets $\mathcal{Z}_v$, $\mathcal{A}_v$ and $\mathcal{C}_v$ can be computed
in $\Oh(2^{\Oh(k)} (n+m))$ time in a straightforward manner. The computation of $\C = \bigcup_{v \in V(G)} \mathcal{C}_v$ in $\Oh(2^{\Oh(k)} n(n+m))$ time can be done by inserting all the elements of $\bigcup_{v \in V(G)} \mathcal{C}_v$ into a prefix tree (trie), each in $\Oh(n)$ time, and ignoring encountered duplicates.
\end{proof}

\begin{definition}[\bf chips touching]
We say that two chips $C_1, C_2 \in \C, C_1 \neq C_2$, {\emph{touch}} each other, denoted $C_1\touch C_2$,
if $C_1 \cap C_2 \neq \emptyset$ or $E(C_1,C_2) \neq \emptyset$.
\end{definition}

The following lemma provides an alternative definition of touching that we will find useful.

\begin{lemma}
\label{lemma:touch}
$C_1 \in \C$ touches $C_2 \in \C$ if and only if $N(C_1) \cap C_2 \neq \emptyset$.
\end{lemma}
\begin{proof}
From right to left, if $v\in N(C_1) \cap C_2$ then there exists a neighbour $u$ of $v$ that belongs to $C_1$, and consequently $uv\in E(C_1,C_2)$.

From left to right, first assume $C_1 \cap C_2 \neq \emptyset$.
Since $\C$ contains only inclusionwise maximal chips, we have that $C_2 \setminus C_1 \neq \emptyset$.
By property~(\ref{def:chip:a}) of Definition~\ref{def:chip} the graph $G[C_2]$ is connected, hence there 
is an edge between $C_2 \setminus C_1$ and $C_1 \cap C_2$ inside $G[C_2]$. This proves $N(C_1) \cap C_2 \neq \emptyset$.

In the other case, assume that $C_1 \cap C_2 = \emptyset$ but there exists $uv \in E(C_1,C_2)$ such that $u\in C_1$ and $v\in C_2$. Since $C_1 \cap C_2 = \emptyset$, it follows that $v\notin C_1$, and hence $v\in N(C_1)\cap C_2$.
\end{proof}

The next result provides the most important tool for bounding the size of adhesions in the constructed decomposition.

\begin{lemma}
\label{lem:touch-bound}
Any chip $C \in \C$ touches at most $3k \cdot 4^{3k}$ other chips of $\C$.
\end{lemma}
\begin{proof}
Assume that $C$ touches some $C' \in \C$. 
By Lemma~\ref{lemma:touch} there exists a vertex $v \in N(C) \cap C'$.
Observe that since $N(C')$ is an important $C'-S$ separator,
then $N(C')$ is also an important $v-S$ separator.
By Lemma~\ref{lem:imp-seps} there are at most $4^{3k}$ important $v-S$
separators of size at most $3k$.
Since $|N(C)| \le 3k$ (by property~(\ref{def:chip:b}) of Definition~\ref{def:chip}), we infer that $C$ touches
at most $3k \cdot 4^{3k}$ chips from $\C$.
\end{proof}

\subsection{Local decomposition}
\label{ssec:decomposition-local}

Equipped with basic properties of chips
we are ready to prove the main step of 
the decomposition part of the paper.
In what follows we show that given a $(2k,k)$-unbreakable set $S$ 
of size bounded in $k$ one can find a (potentially large)
unbreakable part $A \subseteq V$ of the graph, such that $S \subseteq A$
and each connected component of $G \setminus A$ is adjacent
to a small number of vertices of $A$. In what follows, let us define 
\begin{eqnarray*}
\eta & = & 3k \cdot (3k\cdot 4^{3k}+1), \\
\tau & = & (3k)^2\cdot 8^{3k}+2k.
\end{eqnarray*}

\begin{theorem}
\label{thm:decomposition-local}
There is an $\Oh(2^{O(k)} nm)$ time algorithm that, given
a connected graph $G$ together with an integer $k$ and 
a $(2k,k)$-unbreakable set $S \subseteq V(G)$, computes
a set $A \subseteq V(G)$ such that:
\begin{enumerate}[(a)]
  \item \label{item:0} $S \subseteq A$,
  \item \label{item:1} for each connected component $D$ of $G \setminus A$
  we have $|N_G(D)| \le \eta$,
  \item \label{item:2} $A$ is $(\tau,k)$-unbreakable in $G$, and 
  \item \label{item:0b}  if $|S| > 3k$, $G \setminus S$ is connected and $N(V(G) \setminus S) = S$, then $S \neq A$.
\end{enumerate}
\end{theorem}

\begin{proof}
Let $\C$ be the set of inclusionwise maximal chips, enumerated by Lemma~\ref{lemma:chips-time-bound}. Define the set $A$ as follows:
$$A = \left(\bigcap_{C \in \C} V(G) \setminus N[C]\right) \cup \bigcup_{C \in \C} N(C)$$
In the definition we assume that when $\C$ is empty, then $A = V(G)$.
The claimed running time of the algorithm follows directly from Lemma~\ref{lemma:chips-time-bound}.

For property (\ref{item:0}), note that no vertex of $S$ is contained in a chip of $\C$, hence $S \subseteq A$.
We now show property (\ref{item:0b}). Note that $N(V(G) \setminus S) = S$ and $|S| > 3k$ implies $S \neq V(G)$.
Consequently, if $\C = \emptyset$, property (\ref{item:0b}) is straightforward.
Otherwise, let $C \in \C$. Note that $|S| > 3k$ implies that $S \setminus N(C) \neq \emptyset$
and the connectivity of $G \setminus S$ together with $N(V(G)\setminus S)=S$ further implies that $N(C) \setminus S \neq \emptyset$.
Consequently, $A \setminus S \neq \emptyset$ and property (\ref{item:0b}) is proven.

We now move to the remaining two properties.
\begin{claim}
\label{claim:chip-containment}
For any connected component $D$ of $G \setminus A$
there exists a chip $C_1 \in \C$ such that $D \subseteq C_1$.
\end{claim}

\begin{proof}
Observe that a vertex which is not contained in any chip
belongs to the set $A$, as it is either contained in $N(C)$ for some $C \in \C$
or it belongs to $V(G) \setminus N[C]$ for every $C \in \C$.
Let $D$ be an arbitrary connected component of $G \setminus A$
and let $v \in D$ be its arbitrary vertex.
As $v \notin A$, there is a chip $C_v \in \C$ such that $v \in C_v$.
Recall that by its definition the set $A$ contains all the neighbours 
of all the chips in $\C$, hence $N(C_v) \cap D = \emptyset$
and by the connectivity of $G[D]$ we have $D \subseteq C_v$.
\cqed\end{proof}

In the following claim we show that the set $A$ 
satisfies property~(\ref{item:1}) of Theorem~\ref{thm:decomposition-local}.

\begin{lemma}
\label{lem:A-small-neighbourhoods}
For any connected component $D$ of $G \setminus A$ it holds that $|N(D)|\leq \eta$.
\end{lemma}

\begin{proof}
Let $D$ be an arbitrary connected component of $G \setminus A$.
By Claim~\ref{claim:chip-containment} there exists $C \in \C$
such that $D \subseteq C$. 
Intuitively each vertex of $N(D)$ belongs to the set $A$
for one of two reasons: (i) it belongs to $N(C)$,
or (ii) it is adjacent to a vertex of some other chip, which touches $C$.
In both cases we show that there is only a bounded number of such vertices, which is formalized as follows.

Let $v$ be any vertex of $N(D)$. 
Clearly $v \in N[C]$, hence we either have $v \in N(C)$ or $v \in C$. 
Observe that if $v \in C$, then since $v \in A$, by the definition of the set $A$
we have $v \in N(C')$ for some $C' \in \C$, $C'\neq C$. Since $v\in N(C')\cap C$, then $C'$ touches $C$ by Lemma~\ref{lemma:touch}.
We infer that $N(D) \subseteq N(C) \cup \bigcup_{C' \in \C, C\touch C'} N(C')$.
The claimed upper bound on $|N(D)|$ follows from Lemma~\ref{lem:touch-bound}.
\end{proof}

Next, we show that the set $A$ is unbreakable.
A short an informal rationale behind this property
is that everything what could be easily
cut out of the graph was already excluded in the definition of $A$.

\begin{lemma}
\label{lem:A-unbreakable}
The set $A$ is $(\tau,k)$-unbreakable.
\end{lemma}

\begin{proof}
Assume the contrary, and let $(X,Y)$ be a witnessing separation, i.e. we have
that $|X \cap Y| \le k$, $|(X \setminus Y) \cap A|>\tau$ and $|(Y \setminus X) \cap A|>\tau$.
Since $S$ is $(2k,k)$-unbreakable, then either $|(X \setminus Y) \cap S| \le 2k$ or
$|(Y \setminus X) \cap S| \le 2k$.
Without loss of generality we assume that $|(X \setminus Y) \cap S| \le 2k$.
Let us define a set $Q = (X \cap Y) \cup (X \cap S)$ and observe that $|Q| \le 3k$.

Note that each connected component of $G \setminus Q$
is either entirely contained in $X \setminus Y$ or in $Y \setminus X$ (see Fig.~\ref{fig:1}).
Consider connected components of the graph $G \setminus Q$ that are 
contained in $X \setminus Y$ and observe that they
contain at least $|((X \setminus Y) \cap A) \setminus S|>\tau - 2k$
vertices of $A$ in total.
Therefore, by grouping the connected components of $G \setminus Q$ 
contained in $X \setminus Y$ by their neighbourhoods in $Q$,
we infer that there exists a set of connected components 
$\mathcal{D}=\{D_1,\ldots,D_r\}$,
such that $\forall_{1\leq i,j\leq r} N_G(D_i) = N_G(D_j)$ and 
\begin{align}
\label{eq:1}
\left|\bigcup_{i=1}^r D_i \cap A\right| > \frac{\tau-2k}{2^{3k}} = (3k)^2\cdot 4^{3k}\,.
\end{align}

\begin{figure}[htbp!]
        \centering
        \begin{subfigure}[b]{0.45\textwidth}
                \def\svgwidth{\textwidth}
                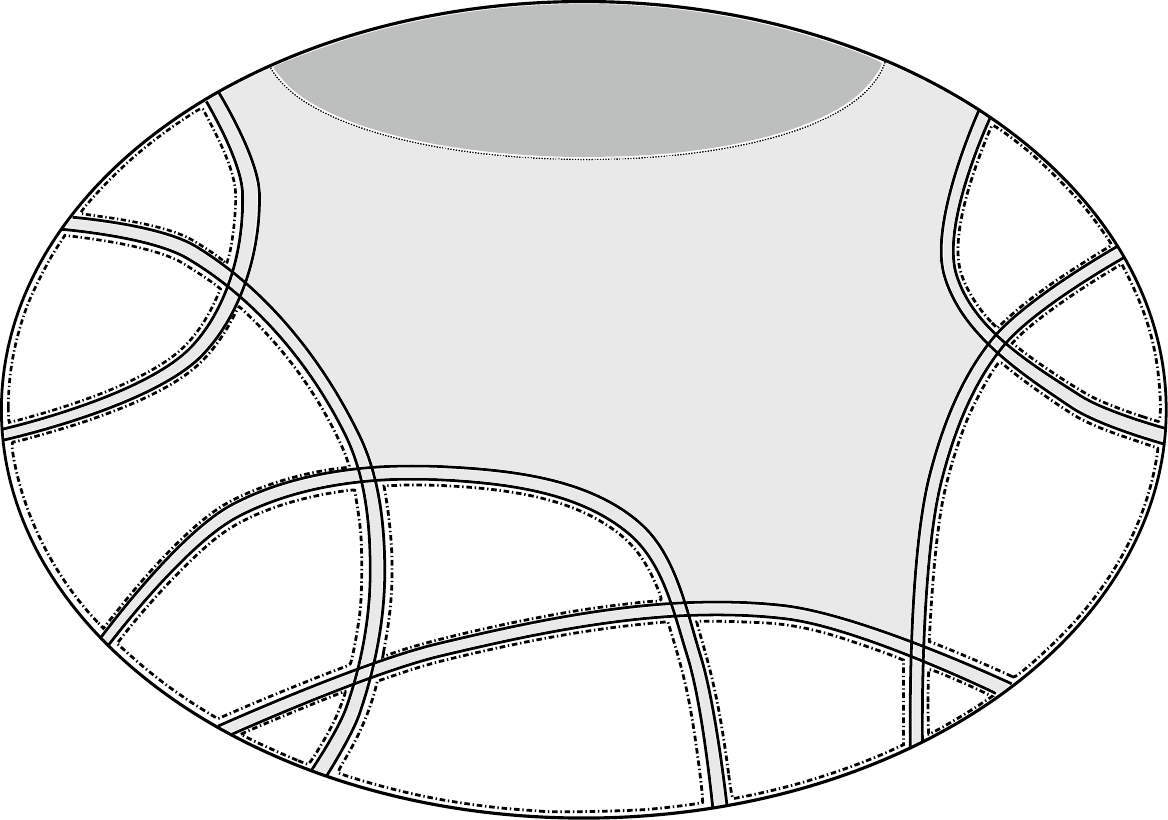
        \caption{Construction of the bag $A$}\label{fig:1}
	\end{subfigure}
        \qquad
        \begin{subfigure}[b]{0.45\textwidth}
                \def\svgwidth{\textwidth}
                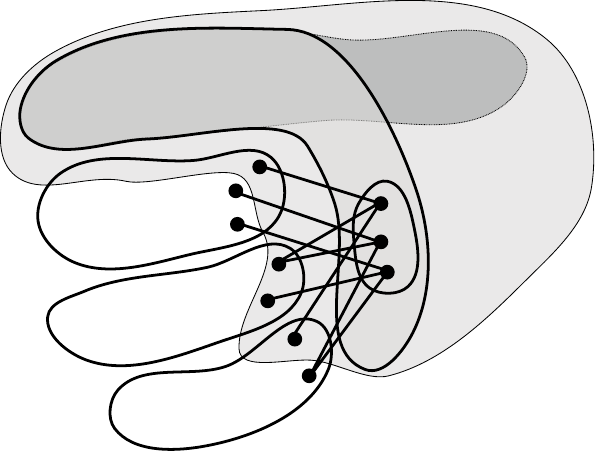
        \caption{Situation in the proof of Claim~\ref{claim:1}.}\label{fig:2}
        \end{subfigure}
\caption{Illustrations of the proof of Theorem~\ref{thm:decomposition-local}}\label{fig:main}
\end{figure}

We now need the following claim.

\begin{claim}
\label{claim:1}
There is a subset $\C_0 \subseteq \C$, such that each $v \in \bigcup_{i=1}^r D_i \cap A$ 
belongs to some chip of $\C_0$, and there are at most $3k \cdot 4^{3k}$ chips in $\C$ that touch some chip of $\C_0$.
\end{claim}

\begin{proof}
Observe that, for each $1 \le i \le r$, $Q$ is a $D_i-S$ separator (see Fig.~\ref{fig:1})
of size at most $3k$.
Therefore, for each $D_i$ there is an important 
$D_i - S$ separator of size at most $3k$ disjoint with $D_i$, hence each $D_i$
is contained in some chip of $\C$. Consider two cases. 

First, assume that for each $1 \le i \le r$ we have $D_i \in \C$.
As $\C_0$ take $\{D_1, \ldots, D_r\}$.
Observe that as components $D_i$ have the same neighbourhoods in $G$,
then by Lemma~\ref{lemma:touch} each chip of $\C$ that touches some chip $D_i$ touches
also $D_1$. Therefore, by Lemma~\ref{lem:touch-bound} 
there are at most $3k \cdot 4^{3k}$ chips in $\C$ that touch some chip of $\C_0$.

In the second case assume that there exist $1 \le i_0 \le r$ and a chip $C \in \C$ such that $D_{i_0} \subsetneq C$.
We shall prove that for each $1 \le i \le r$ we have $D_i \subseteq C$.
Since $C$ is connected and $C\setminus D_{i_0}$ is non-empty, we have that $C \cap N(D_{i_0}) \neq \emptyset$. 
Let $C' = C \cup \bigcup_{1 \le i \le r} D_i$. 
Clearly $C' \cap S = \emptyset$, and $C'$ is connected since each component $D_i$ is adjacent to every vertex of $C\cap N(D_{i_0})$.
Moreover, as each $D_i$ has the same neighbourhood in $Q$ we have $|N(C')| \le |N(C)| \le 3k$ (see Fig.~\ref{fig:2}).
As $\C$ contains only maximal chips we have $C' = C$ and hence $\bigcup_{1 \le i \le r} D_i \subseteq C$.
Define $\C_0$ as $\{C\}$. 
By Lemma~\ref{lem:touch-bound} a single chip touches at most $3k \cdot 4^{3k}$ other chips, which finishes the proof of Claim~\ref{claim:1}
\cqed\end{proof}

Let $v \in A \cap D_i$ for some $1 \le i \le r$.
Since $v$ is contained in some $C' \in \C_0$, we have $v \notin V(G) \setminus N[C']$. 
Consequently, by the definition of the set $A$ there exists a chip $C_v \in \C$ such that $v \in N(C_v)$.
Note that $C' \neq C_v$ and $N(C_v) \cap C' \neq \emptyset$, hence by Lemma~\ref{lemma:touch} $C'$ touches $C_v$.
By Claim~\ref{claim:1} there are at most $3k \cdot 4^{3k}$ chips touching a chip of $\C_0$.
As each $C_v$ satisfies $|N(C_v)| \le 3k$, we infer that the number of vertices
of $A$ in $\bigcup_{1 \le i \le r} D_i$ is at most $(3k)^2 4^{3k}$, which contradicts $(\ref{eq:1})$
and finishes the proof of Lemma~\ref{lem:A-unbreakable}.
\end{proof}

Lemma~\ref{lem:A-small-neighbourhoods} and Lemma~\ref{lem:A-unbreakable} ensure properties (\ref{item:1}) and (\ref{item:2}) of the set $A$, respectively. 
This concludes the proof of Theorem~\ref{thm:decomposition-local}.
\end{proof} 

\subsection{Strengthening unbreakability of adhesions}
\label{ssec:strengthen}

So far Theorem~\ref{thm:decomposition-local} provides us with a construction of the bag that meets almost all the requirements, apart from $(2k,k)$-unbreakability of adhesions. For this reason, in this section we want to show that the set $A$ from Theorem~\ref{thm:decomposition-local} can be extended to a set $A'$
in such a way that for each connected component $D$ of $G\setminus A'$
the set $N_G(D)$ is even $(2k,k)$-unbreakable. During this extension we may weaken unbreakability of $A'$, but if we are careful enough then this loss will be limited to a single-exponential function of $k$.  We start with the following recursive procedure.

\begin{lemma}
\label{lem:breaking-neighbourhood}
Let $G$ be a graph, and $L \subseteq V(G)$ be a subset of vertices of size at least $2k+1$.
Then one can in $\Oh(|L|^{4k+3} kn(n+m))$ time find a set $L'$, $L\subseteq L'$, such that 
$|L' \setminus L| \le (|L|-2k-1) \cdot k$
and for each connected component $D$ of $G \setminus L'$, we have that $|N_G(D)|\leq |L|$ and $N_G(D)$ is $(2k,k)$-unbreakable in $G$.
\end{lemma}
\begin{proof}
We prove the lemma by induction on $|L|$, with the following two base cases. If $L = V(G)$, clearly we may return $L' = L$.
In the second base case we assume that $L$ is $(2k,k)$-unbreakable in $G$, which can be checked in $\Oh(|L|^{4k+2} k(n+m))$ time using Lemma~\ref{lem:breaking}.
Then for each connected component $D$ of $G\setminus L$ we have that $N_G(D)\subseteq L$, and thus $|N_G(D)|\leq |L|$ and $N_G(D)$ is also $(2k,k)$-unbreakable in $G$. 
Hence we can set $L'=L$, and since $|L| \ge 2k+1$, we have that $|L' \setminus L| \le (|L|-2k-1) \cdot k$.

Now let us assume that $L$ is $(2k,k)$-breakable in $G$, and hence there exists a separation $(X,Y)$ of $G$ such that
$|X\cap Y|\leq k$, $|(X \setminus Y) \cap L| > 2k$ and $|(Y \setminus X) \cap L| > 2k$, found by the algorithm of Lemma~\ref{lem:breaking}.
We use inductively Lemma~\ref{lem:breaking-neighbourhood}
for the pair $(G_1 = G[X],L_1 = (X \cap L) \cup (X \cap Y))$
and for the pair $(G_2 = G[Y],L_2 = (Y \cap L) \cup (X \cap Y))$,
to obtain sets $L_1'$ and $L_2'$, respectively. Note here that $|L_1|,|L_2| \ge 2k+1$ and $|L_1|,|L_2| < |L|$.
Define $L' = L_1' \cup L_2'$.
Each connected component $D$ of $G \setminus L'$
is either a connected component of $G_1 \setminus L_1'$ and is adjacent only to $L_1'$, or is a connected component of $G_2 \setminus L_2'$ and is adjacent only to $L_2'$.
Assume without of loss of generality the first case. By inductive assumption we infer that $|N_{G_1}(D)|\leq |L_1|$ and $N_{G_1}(D)$ is $(2k,k)$-unbreakable in $G_1$, and since $N_{G_1}(D)=N_G(D)$, $|L_1|<|L|$, and $G_1$ is a subgraph of $G$, then it follows that $|N_G(D)|\leq |L|$ and $N_G(D)$ is $(2k,k)$-unbreakable in $G$.
It remains to argue that the cardinality of $L' \setminus L$ is not too large.
Observe that
$$L' \setminus L \subseteq (L_1' \setminus L_1) \cup (L_2' \setminus L_2) \cup (X \cap Y)\,;$$
therefore, by induction we have 
\begin{align*}
|L'\setminus L| & \le (|L_1|-2k-1) \cdot k + (|L_2|-2k-1) \cdot k + k \\
    & \le (|L_1|+|L_2|-4k-1)\cdot k \\
    &  \le (|L|+2|X\cap Y|-4k-1)\cdot k \\
    & \le (|L|-2k-1)\cdot k\,.
\end{align*}
Let us now bound the running time of the recursion. Clearly, as the size of the set $L$ decreases in the recursive calls,
the depth of the recursion is at most $|L|$.
Moreover, note that any vertex may appear in $V(G) \setminus L$ in at most one recursive call $(G,L)$ at any fixed level of the recursion tree.
Hence, there are at most $|L|n$ recursive calls that do not correspond to the first base case, and, consequently, at most $2|L|n+1$ recursive calls in total.
As each recursive call takes $\Oh(|L|^{4k+2} k(n+m))$ time, the promised running time bound follows.
\end{proof}

We can now proceed to strengthen Theorem~\ref{thm:decomposition-local} by including also the procedure of Lemma~\ref{lem:breaking-neighbourhood}. In the following, let
\begin{eqnarray*}
\tau' & = & \tau+\left(\binom{\tau+k}{2} \cdot k +k\right) \cdot k\eta\,.
\end{eqnarray*} 

\begin{theorem}
\label{thm:decomposition-local-enhanced}
There is an $\Oh(2^{O(k^2)} nm)$ time algorithm that,
given a connected graph $G$ together with an integer $k$ 
and a $(2k,k)$-unbreakable set $S$, computes a set $A' \subseteq V(G)$
such that:
\begin{enumerate}[(a)]
  \item $S \subseteq A'$,
  \item \label{item:11} for each connected component $D$ of $G \setminus A'$
  the set $N_G(D)$ is $(2k,k)$-unbreakable, and $|N_G(D)| \leq \eta$,
  \item \label{item:22} $A'$ is $(\tau',k)$-unbreakable in $G$,
  \item \label{item:00b} moreover, if $|S| > 3k$, $G \setminus S$ is connected and $N(V(G) \setminus S) = S$, then $S \neq A'$.
\end{enumerate}
\end{theorem}

\begin{proof}
We start by finding the set $A$ by running the algorithm  Theorem~\ref{thm:decomposition-local}.
Next, for each connected component $D$ of $G\setminus A$
 using Lemma~\ref{lem:breaking} we check
whether $N(D)$ is $(2k,k)$-breakable in $G$.
By Theorem~\ref{thm:decomposition-local}, the cardinality
of $N(D)$ is bounded by $\eta$, hence
all tests take total time $\Oh(\eta^{4k+2} knm) = \Oh(2^{\Oh(k^2)} nm)$ time.
Note that if $N(D)$ is $(2k,k)$-breakable in $G$, then
in particular $|N(D)| > 2k$, hence we can use
Lemma~\ref{lem:breaking-neighbourhood} 
for the pair $(G[N[D]],L_D = N(D))$; let $L_D'$ be the obtained set.
As $|L_D| \leq \eta$, the algorithm of Lemma~\ref{lem:breaking-neighbourhood} runs in $\Oh(\eta^{4k+3}k |N[D]|m)$ time for a fixed component $D$,
and  total time taken by calls to Lemma~\ref{lem:breaking-neighbourhood} is:
$$\sum_D \Oh(\eta^{4k+3} k(|D| + |N(D)|)m) \leq \Oh(\eta^{4k+3}km) \cdot \left(\sum_D |D| + \sum_D \eta\right) = \Oh(\eta^{4k+4} knm) = \Oh(2^{\Oh(k^2)} nm).$$
In the case when $N(D)$ is $(2k,k)$-unbreakable, let $L_D = L_D' = N(D)$.
Define $A' = A \cup (\bigcup_{D} L_D')$, where
the union is taken over all the connected components $D$ of $G \setminus A$.

Since $S \subseteq A \subseteq A'$, we have that $S \subseteq A'$, and, moreover, the property (\ref{item:00b}) follows directly from property (\ref{item:0b}) of Theorem~\ref{thm:decomposition-local}.
Moreover, as $|L_D|\leq \eta$ for
each connected component $D$ of $G\setminus A$, then by Lemma~\ref{lem:breaking-neighbourhood} for each connected component $D'$ of $G \setminus A'$ we also have $|N_G(D')| \le \eta$.
The fact that $N_G(D')$ is $(2k,k)$-unbreakable in $G$ follows directly from Lemma~\ref{lem:breaking-neighbourhood}. It remains to show that $A'$ is $(\tau',k)$-unbreakable in $G$.

Consider any separation $(X,Y)$ of $G$ of order at most $k$.
By Theorem~\ref{thm:decomposition-local} the set $A$ is $(\tau,k)$-unbreakable,
hence either $|(X \setminus Y) \cap A| \le \tau$ or 
$|(Y \setminus X) \cap A| \le \tau$, and without loss of generality assume the former.
As $(X,Y)$ is an arbitrary separation of order at most $k$,
to show that $A'$ is $(\tau',k)$-unbreakable
it suffices to prove that $|(X \setminus Y) \cap (A' \setminus A)| \le (\binom{\tau+k}{2} \cdot k +k) \cdot k\eta$.

Note that $A' \setminus A \subseteq \bigcup_{D} L_D' \setminus L_D$.
As for each $D$ we have $|L_D' \setminus L_D| \le k\eta$ by Lemma~\ref{lem:breaking-neighbourhood},
to finish the proof of Theorem~\ref{thm:decomposition-local-enhanced}
we are going to show that there are at most $\binom{\tau+k}{2} \cdot k+k$
connected components $D$ of $G \setminus A$ such that 
$D \cap (X\setminus Y) \neq \emptyset$ and $L_D' \neq L_D$.
As $(X,Y)$ is of order at most $k$, there
are at most $k$ connected components $D$ of $G \setminus A$
intersecting $X \cap Y$.
Hence we restrict our attention to connected components $D$
of $G\setminus A$, such that $D \subseteq X \setminus Y$,
which in turn implies $N(D) \subseteq A \cap X$.
Recall that if $L_D' \neq L_D$ for such a connected component $D$,
then $N(D)$ is $(2k,k)$-breakable in $G$, and
hence there exist two vertices $v_a,v_b \in N(D) \subseteq A \cap X$,
such that the minimum vertex cut separating $v_a$ and $v_b$
in $G$ is at most $k$.
However, such a pair of vertices $v_a,v_b$ may be simultaneously
contained in neighbourhoods of at most $k$ connected components $D$, since each component $D$ adjacent both to $v_a$ and to $v_b$ contributes with at least one path between them.
As $|A \cap X| \leq \tau + k$, the theorem follows.
\end{proof}

\subsection{Constructing a decomposition} 
\label{ssec:decomposition-main}

In this subsection we prove our main decomposition theorem, i.e. Theorem~\ref{thm:decomposition-main}.
However, for the inductive approach to work we need a bit stronger statement,
where additionally we have a set $S \subseteq V(G)$ that has to be contained
in the top bag of the tree decomposition.
Note that Theorem~\ref{thm:decomposition-main} follows from the following  
by setting $S=\emptyset$.

\begin{theorem}
\label{thm:decomposition-recursion}
There is an $\Oh(2^{\Oh(k^2)}n^2m)$ time algorithm
that, given a connected graph $G$ together with
an integer $k$ and a set $S \subseteq V(G)$ of size at most $\eta$ such that $G \setminus S$ is connected and $N(V(G) \setminus S) = S$,
computes a tree decomposition $(T,\beta)$ such that $S$ is contained 
in the top bag of the tree decomposition, 
and  the following conditions
are satisfied:
\begin{enumerate}[(i)]
  \item\label{p:r:1} for each $t \in V(T)$, the graph $G[\gamma(t)]\setminus \sigma(t)$ is connected and $N(\gamma(t)\setminus \sigma(t))=\sigma(t)$,
  \item\label{p:r:2} for each $t \in V(T)$, the set $\beta(t)$ is $(\tau',k)$-unbreakable in $G[\gamma(t)]$,
  \item\label{p:r:3} for each non-root $t \in V(T)$, we have that $|\sigma(t)|\leq \eta$ and $\sigma(t)$ is $(2k,k)$-unbreakable in $G[\gamma(\parent(t))]$.
  \item\label{p:r:4} $|V(T)| \leq |V(G) \setminus S|$.
\end{enumerate}
\end{theorem}

\begin{proof}

If $|V(G)| \leq \tau'$, the algorithm creates a single bag containing the entire $V(G)$.
It is straightforward to verify that such a decomposition satisfies all the required properties.
Thus, in the rest of the proof we assume that $|V(G)| > \tau'$, in particular, $|V(G)| > 3k$.

Define $S' = S$ and, if $|S| \leq 3k$, add $3k+1-|S|$ arbitrary vertices of $V(G) \setminus S$ to $S'$. Note that,
as $\eta > 3k$, we have $3k < |S'| \leq \eta$.

We now define a set $A'$ as follows. 
First, we verify, using Lemma~\ref{lem:breaking}, whether $S'$ is $(2k,k)$-breakable in $G$ or not.
If it turns out to be $(2k,k)$-breakable in $G$, we apply Lemma~\ref{lem:breaking-neighbourhood}
to the pair $(G,S')$, obtaining a set which we denote by $A'$.
Otherwise, we can use Theorem~\ref{thm:decomposition-local-enhanced} on the pair $(G,S')$
to obtain a set $A'$.
Note that in both cases $S \subseteq S' \subseteq A'$
and all computations so far take $\Oh(2^{\Oh(k^2)}nm)$ time in total.

Regardless of the way the set $A'$ was obtained, we proceed with it as follows.
For each connected component $D$ of $G \setminus A'$, we use Theorem~\ref{thm:decomposition-recursion} inductively
for the graph $G[N[D]]$ and $S_D = N(D)$.
Let us now verify that  (a) each $S_D$ is $(2k,k)$-unbreakable in $G$, (b) that the assumptions of the theorem are satisfied,
and (c) that the recursive call is applied to a strictly smaller instance in the sense defined in the following.


For the first two claims, if $S$ is $(2k,k)$-breakable, Lemma~\ref{lem:breaking-neighbourhood} asserts that $|S_D| \leq |S| \leq \eta$ and $S_D$ is $(2k,k)$-unbreakable in $G$.
Otherwise, property~(\ref{item:11}) of Theorem~\ref{thm:decomposition-local-enhanced} ensures that $|S_D| \leq \eta$ and $S_D$ is $(2k,k)$-unbreakable in $G$. The other assumptions on the set $S_D$ in the recursive calls follow directly from the definitions of these calls.

For the last claim, we show that either $|N[D]| < |V(G)|$ or $N[D] = V(G)$ and $|D| < |V(G) \setminus S|$.
Assume the contrary, that is, $D = V(G) \setminus S$ and $N(D) = S_D = S = S' = A'$.
In particular, as $S_D$ is $(2k,k)$-unbreakable in $G$, the set $A'$ was obtained using Theorem~\ref{thm:decomposition-local-enhanced}.
However, as $|S'| > 3k$, property (\ref{item:00b}) of Theorem~\ref{thm:decomposition-local-enhanced}
ensures that $S' \subsetneq A'$, a contradiction.

Let $(T_D,\beta_D)$ be the tree decomposition obtained in the recursive call for the pair $(G[N[D]],S_D)$.
Construct a tree decomposition $(T,\beta)$, 
by creating an auxiliary node $r$, which will be the root of $T$,
and attach $T_D$ to $r$, by making the root $r_D$ of $T_D$
a child of $r$ in $T$.
Finally, define $\beta = \bigcup_{D} \beta_D$ and set $\beta(r)=A'$.
A straightforward check shows that $(T,\beta)$ is indeed a valid tree decomposition.
We now proceed to verify its promised properties.

Clearly, $S \subseteq S' \subseteq A'$. For any connected component $D$ of $G \setminus A'$, note
that $\gamma(r_D) = N[D]$ and $\sigma(r_D) = N(D) = S_D$. 
This, together with inductive assumptions on recursive calls, proves properties (\ref{p:r:1}) and (\ref{p:r:3}).

If $A'$ is obtained using Lemma~\ref{lem:breaking-neighbourhood}, then $|A'| \leq k|S'| \leq k\eta < \tau'$, hence
clearly $A' = \beta(r)$ is $(\tau',k)$-unbreakable. 
In the other case, property~(\ref{item:22}) of Theorem~\ref{thm:decomposition-local-enhanced}
ensures the unbreakability promised in property (\ref{p:r:2}).

It remains to bound the number of bags of $(T,\beta)$; as each bag is processed in $\Oh(2^{\Oh(k^2)}nm)$ time this would also prove the promised running time bound.
Note that by property (\ref{p:r:4}) for the recursive calls we have that $|V(T_D)| \leq |D|$ and, consequently, 
$|V(T)| \leq |V(G) \setminus A'| + 1 = |V(G) \setminus S| + 1 - |A' \setminus S|$.
To finish the proof of property (\ref{p:r:4}) it suffices to show that $S \subsetneq A'$.
If $S \subsetneq S'$, the claim is straightforward.
Otherwise, if $S=S'$ is $(2k,k)$-breakable, then Lemma~\ref{lem:breaking-neighbourhood} cannot return $A' = S'$ as $G \setminus S'$ is connected
and $N(V(G) \setminus S') = S'$ is not $(2k,k)$-unbreakable. Consequently, $S' \subsetneq A'$ in this case.
In the remaining case, when $S=S'$ is $(2k,k)$-unbreakable, property (\ref{item:00b}) of Theorem~\ref{thm:decomposition-local-enhanced}
ensures that $S' \subsetneq A'$.
This finishes the proof of Theorem~\ref{thm:decomposition-recursion}.
\end{proof}

%% file: decomp.pdf_tex
\begingroup%
  \makeatletter%
  \providecommand\color[2][]{%
    \errmessage{(Inkscape) Color is used for the text in Inkscape, but the package 'color.sty' is not loaded}%
    \renewcommand\color[2][]{}%
  }%
  \providecommand\transparent[1]{%
    \errmessage{(Inkscape) Transparency is used (non-zero) for the text in Inkscape, but the package 'transparent.sty' is not loaded}%
    \renewcommand\transparent[1]{}%
  }%
  \providecommand\rotatebox[2]{#2}%
  \ifx\svgwidth\undefined%
    \setlength{\unitlength}{336.275bp}%
    \ifx\svgscale\undefined%
      \relax%
    \else%
      \setlength{\unitlength}{\unitlength * \real{\svgscale}}%
    \fi%
  \else%
    \setlength{\unitlength}{\svgwidth}%
  \fi%
  \global\let\svgwidth\undefined%
  \global\let\svgscale\undefined%
  \makeatother%
  \begin{picture}(1,0.70198249)%
    \put(0,0){\includegraphics[width=\unitlength]{decomp.pdf}}%
    \put(0.49885454,0.62948275){\color[rgb]{0,0,0}\makebox(0,0)[lb]{\smash{$S$}}}%
  \end{picture}%
\endgroup%

%% file: comp-Q.pdf_tex
\begingroup%
  \makeatletter%
  \providecommand\color[2][]{%
    \errmessage{(Inkscape) Color is used for the text in Inkscape, but the package 'color.sty' is not loaded}%
    \renewcommand\color[2][]{}%
  }%
  \providecommand\transparent[1]{%
    \errmessage{(Inkscape) Transparency is used (non-zero) for the text in Inkscape, but the package 'transparent.sty' is not loaded}%
    \renewcommand\transparent[1]{}%
  }%
  \providecommand\rotatebox[2]{#2}%
  \ifx\svgwidth\undefined%
    \setlength{\unitlength}{171.12700185bp}%
    \ifx\svgscale\undefined%
      \relax%
    \else%
      \setlength{\unitlength}{\unitlength * \real{\svgscale}}%
    \fi%
  \else%
    \setlength{\unitlength}{\svgwidth}%
  \fi%
  \global\let\svgwidth\undefined%
  \global\let\svgscale\undefined%
  \makeatother%
  \begin{picture}(1,0.75865841)%
    \put(0,0){\includegraphics[width=\unitlength]{comp-Q.pdf}}%
    \put(0.16874405,0.36236246){\color[rgb]{0,0,0}\makebox(0,0)[lb]{\smash{$D_1$}}}%
    \put(0.16214659,0.21201868){\color[rgb]{0,0,0}\makebox(0,0)[lb]{\smash{$D_2$}}}%
    \put(0.23460853,0.05367588){\color[rgb]{0,0,0}\makebox(0,0)[lb]{\smash{$D_3$}}}%
    \put(0.73647489,0.62263651){\color[rgb]{0,0,0}\makebox(0,0)[lb]{\smash{$S$}}}%
    \put(0.54861043,0.49381508){\color[rgb]{0,0,0}\makebox(0,0)[lb]{\smash{$Q$}}}%
    \put(0.72842338,0.34083982){\color[rgb]{0,0,0}\makebox(0,0)[lb]{\smash{$N(D_i)$}}}%
    \put(0.89481755,0.50186646){\color[rgb]{0,0,0}\makebox(0,0)[lb]{\smash{$A$}}}%
  \end{picture}%
\endgroup%

%% file: dp.tex
\newcommand{\duzo}{\infty}
\newcommand{\accum}{\mu}

In this section we show a dynamic programming routine
defined on the decomposition given by Theorem~\ref{thm:decomposition-main}.
When handling one bag of the decomposition, we essentially follow the approach
of the high connectivity phase of ``randomized contractions''~\cite{rand-contractions}.
That is, we apply the colour-coding technique in a quite involved fashion, to
highlight the solution in a bag (relying heavily on the unbreakability of the bag),
and then we analyse the outcome by a technical, but quite natural knapsack-style
dynamic programming.

The section is organized as follows.
First, in Section~\ref{ssec:bisection-problem} we define an abstract
problem which encapsulates the computational 
task one needs to perform in a single step of the dynamic
programming procedure.
This one, in turn, is presented in Section~\ref{ssec:bisection-dp}.

Through this section we mostly ignore the study of factors polynomial in the graph
size in the running time of the algorithm, and we use the $\Ohstar(\cdot)$ notation.
We do not optimize the exponent of the polynomial in this dependency, as it adds unnecessary level of technicalities to the description,
distracting from the main points of the reasoning,
and, most importantly, is in fact less 
relevant to the main result of this paper --- the fixed-parameter
tractability of \textsc{Minimum Bisection}.
In Section~\ref{ssec:bisection-ptime} we shortly argue how to obtain the running
time promised in Theorem~\ref{thm:bisectionFPT}.

\subsection{Hypergraph painting}
\label{ssec:bisection-problem}

\defproblemG{\hpname\footnote{We are intentionally not using the name {\sc Hypergraph Colouring},
as it has an established, and different, meaning.} (\hp)}{
  Positive integers $k, b, d, q$, a multihypergraph $H$ with hyperedges of size at most $d$,
  a partial function $\col_0 : V(H) \nrightarrow \{\B, \W\}$, 
  and a function $f_\hedge : \{\B,\W\}^\hedge \times \{0,\ldots,b\} \to \{0,1,\ldots,k,\duzo\}$ for each $\hedge \in E(H)$.
}{ For each $0 \le \accum \le b$, compute the value $w_{\accum}$,
    $$w_\accum = \min_{\col \supseteq \col_0, (a_\hedge)_{\hedge \in E(H)}} \sum_{\hedge \in E(H)} f_\hedge(\col|_\hedge,a_\hedge)\,,$$
  where the minimum is taken over colourings $\col : V(H) \to \{\B,\W\}$ extending $\col_0$
  and partitions of $\accum$ into non-negative integers $\accum = \sum_{\hedge \in E(H)}a_\hedge$,
  and the sum attains value $\duzo$ whenever its value exceeds $k$.
}

We denote $n = |V(H)|$ and $m=|E(H)|$ throughout the analysis of the \hpname{} problem.

We call an instance $(k,b,d,q,H,\col_0,(f_\hedge)_{\hedge \in E(H)})$
a {\em proper} instance of \hpname if the following conditions hold:
\begin{itemize}
 \item ({\bf local unbreakability}), 
    for each $\hedge \in E(H)$, each $\col : \hedge \to \{\B,\W\}$
    marking more than $3k$ vertices of each colour, i.e. $|\col^{-1}(\B)|, |\col^{-1}(\W)| > 3k$, 
    and each $0 \le \accum \le b$ the value $f_\hedge(\col, \accum)$ equals $\duzo$,
  \item ({\bf connectivity}), 
    for each $\hedge \in E(H)$, each $\col : \hedge \to \{\B,\W\}$
    marking at least one vertex with each colour, i.e. $|\col^{-1}(\B)|, |\col^{-1}(\W)| > 0$, 
    and each $0 \le \accum \le b$ the value $f_\hedge(\col, \accum)$ is non-zero,
  \item ({\bf global unbreakability})
  for each $0 \le \accum \le b$ such that $w_\accum < \duzo$ there is a 
  witnessing colouring $\col : V(H) \to \{\B,\W\}$,
  which colours at most $q$ vertices with one of the colours, i.e. $\min(|\col^{-1}(\B)|, |\col^{-1}(\W)|) \le q$.
\end{itemize}
Note that, by local unbreakability, for proper instances each function $f_{\hedge}$
can be represented by at most $(2 \sum_{i=0}^{3k} d^i) \cdot (b+1) \le 4(b+1)d^{3k}$ values which are smaller than $\duzo$.

We are going to use the well-established tool of fixed parameter tractability,
namely the colour-coding technique of Alon, Yuster and Zwick~\cite{color-coding}.
A standard method of derandomizing the colour-coding technique
is to use {\em splitters} of Naor et al.~\cite{naor-schulman-srinivasan-derandom}.
We present our algorithm already in its derandomized form,
and for this reason we use the following abstraction 
of splitters.

\newcommand{\randfamily}{\ensuremath{\mathcal{F}}}

\begin{lemma}[Lemma I.1 of~\cite{rand-contractions}]
\label{lem:splitters}
Given a set $U$ of size $n$, and integers $0 \leq a,b \leq n$, one
can in $O(2^{O(\min(a,b) \log (a+b))} n \log n)$ time
construct a family $\randfamily$ of at most $O(2^{O(\min(a,b) \log (a+b))} \log n)$
subsets of $U$, such that the following holds:
for any sets $A,B \subseteq U$, $A \cap B = \emptyset$, $|A|\leq a$, $|B|\leq b$,
there exists a set $S \in \randfamily$ with $A \subseteq S$ and $B \cap S = \emptyset$.
\end{lemma}

In our dynamic programming routine we will use the following operators,
to make the description of the actual algorithm more concise.

\begin{definition}
For two functions $g,h : \{0, \ldots, b\} \to \{0,1,\ldots,k,\duzo\}$
we define functions $g\oplus h,\min(g,h)$ as follows:
\begin{align*}
(g\oplus h)(\accum) & = \min_{\accum_1 + \accum_2 = \accum} g(\accum_1) + h(\accum_2)\,, \\
\min(g,h)(\accum) & = \min(g(\accum), h(\accum)),
\end{align*}
where each integer larger than $k$ is treated as $\duzo$.
\end{definition}

Note that given two functions $g,h$ one can compute $g \oplus h$ in $\Oh(b^2)$ time, and $\min(g,h)$ in $\Oh(b)$ time.

\begin{lemma}
\label{lem:hp}
There is an $\Ohstar(q^{\Oh(k)} \cdot d^{O(k^2)})$ time algorithm solving the \hpname problem for proper instances.
\end{lemma}

\begin{proof}
First, let us fix the value of $\accum$, $0 \le \accum \le b$.
Our goal is to compute a single value $w_\accum$\footnote{Actually our algorithm after a minor modification
computes all the values $w_\accum$ at once, however for the sake of simplicity
we focus on a single value of $\accum$, at the cost of higher polynomial factor.}.
It is enough to compute the correct value of $w_\accum$
assuming 
\begin{align}
\label{eq:hp:0}
w_\accum < \duzo\,.
\end{align}
Consequently, by the global unbreakability property let $\col_{opt} : V(H) \to \{\B,\W\}$ 
be a colouring witnessing the value $w_\accum$, that colours at most $q$
vertices with one of the colours.
Without loss of generality let us assume that 
\begin{align}
\label{eq:hp:1}
|\col_{opt}^{-1}(\W)| \le q\,,
\end{align}
as the other case $|\col^{-1}(\B)| \le q$ is symmetric.

Observe that by the local unbreakability property for
each $\hedge \in E(H)$ there are at most $\ell=4d^{3k}=d^{O(k)}$ possible 
colourings leading to a value of $f_\hedge$ which is different than $\duzo$.
For each $\hedge\in E(H)$ let us order the possible bichromatic colourings of $\hedge$ arbitrarily, and
for $1 \le i \le \ell$ let $\col_{\hedge,i}$ be the $i$-th 
of the possible colouring which is bichromatic on $\hedge$
(if the number of such colourings is smaller than $\ell$ we append the sequence with arbitrary bichromatic colourings).

We want to assign each $\hedge \in E(H)$ to be in one of the following states:
\begin{itemize}
  \item $\hedge$ is definitely monochromatic,
  \item $\hedge$ is either monochromatic, 
  or should be coloured as in $\col_{\hedge,i}$ for a fixed $1 \le i \le \ell$.
\end{itemize}
Formally, for an assignment $p : E(H) \to \{0,\ldots,\ell\}$ 
by $p(\hedge)=0$ we express the ``definitely monochromatic'' state,
and by $p(\hedge)=i > 0$ we express the ``either monochromatic or $i$-th type of bichromatic colouring'' state.

Let $E_{white} = \{\hedge \in E(H) : \col_{opt}(\hedge) = \{\W\}\}$ be the multiset of monochromatic edges
of $E(H)$ coloured all white with respect to $\col_{opt}$.
Moreover let $E_0 \subseteq E_{white}$ be any spanning forest of the hypergraph $(V(H),E_{white})$.
By (\ref{eq:hp:1}) we have $|E_0| \le q$.
Let $E_1 \subseteq E(H)$ be the set of 
edges which are bichromatic with respect to $\col_{opt}$.
Note that by the the connectivity property together
with (\ref{eq:hp:0}) we have $|E_1| \le k$.

We call an assignment $p$ {\emph{good}}, with respect to $\col_{opt}$,
if:
\begin{itemize}
  \item for each $\hedge \in E_1$ we have $p(\hedge)>0$ and $\col_{opt}|_\hedge = \col_{\hedge,p(\hedge)}$,
  \item for each $\hedge \in E_0$ we have $p(\hedge) = 0$.
\end{itemize}

Let $\randfamily$ be a family constructed by the algorithm
of Lemma~\ref{lem:splitters} for the universe $E(H)$
and integers $k,q$ in $\Ohstar(q^{O(k)})$ time.
By the properties of $\randfamily$ there exists
$S_0 \in \randfamily$ such that $E_1 \subseteq S_0$ and $E_0 \cap S_0 = \emptyset$.
We iterate through all possible $S\in \randfamily$; in one of the cases we have $S=S_0$.

In the second level of derandomization we use 
the standard notion of perfect families.
An $(N,r)$-{\em{perfect family}} is a family 
of functions from $\{1,2,\ldots,N\}$ to
$\{1,2,\ldots,r\}$,
such that for any subset $X \subseteq \{1,2,\ldots,N\}$ of size $r$, one of the functions
in the family is injective on $X$. Naor et al.~\cite{naor-schulman-srinivasan-derandom}
gave an explicit construction of an $(N,r)$-perfect family of size $\Oh(e^r r^{\Oh(\log r)} \log N)$
using $\Oh(e^r r^{\Oh(\log r)} N \log N)$ time.
We construct a $(|S|,k)$-perfect family $\mathcal{D}$ of size $\Oh(2^{\Oh(k)} \log |S|)$.
Assuming that we consider the case when $S=S_0$, there exists a function $\delta_0\in \mathcal{D}$, $\delta_0: S_0 \to \{1,\ldots,k\}$, such that $\delta_0$ is injective on $E_1$.
We iterate through all possible functions $\delta\in \mathcal{D}$; providing that $S=S_0$, in one case we have $\delta=\delta_0$.

Finally, we guess, by trying all $\ell^{\Oh(k)}=d^{\Oh(k^2)}$ possibilities,
a function $\delta':\{1,\ldots,k\} \to \{1,\ldots,\ell\}$.
In the case where $S=S_0$ and $\delta=\delta_0$, for at least one such $\delta'$ we have that $\delta'(\delta_0(\hedge))=i$ holds for each $\hedge \in E_1$,
where $\col_{opt}|_\hedge=\col_{\hedge,i}$. 
Summing up, in one of the $\Ohstar(q^{\Oh(k)} \cdot d^{\Oh(k^2)})$ cases
we will end up having a good assignment $p : E(H) \to \{0, \ldots, \ell\}$ at hand.


For an assignment $p : E(H) \to \{0,\ldots,\ell\}$
define an auxiliary undirected simple graph $L_p$,
with a vertex set $V(L) = V(H)$.
For each edge $\hedge \in E(H)$ such that $p(\hedge)=0$
make $\hedge$ a clique in $L_p$.
For each edge $\hedge \in E(H)$ such that $p(\hedge) = i > 0$
make the sets $\col_{\hedge,i}^{-1}(\B)$, $\col_{\hedge,i}^{-1}(\W)$
cliques in $L_p$.

\begin{claim}
\label{claim:hp:1}
If $p$ is a good assignment, 
then all the vertices contained in the same connected
component of $L_p$ are coloured with the same colour
in $\col_{opt}$.
\end{claim}

\begin{proof}
Follows directly from the assumption that $p$ is a good
assignment and from the definition of the graph $L_p$.
\end{proof}

\begin{claim}
\label{claim:hp:2}
Let $p$ be a good assignment.
If $D$ is a connected component of $L_p$
coloured white by $\col_{opt}$,
then each edge $\hedge \in E(H)$ 
such that $\hedge \cap D \neq \emptyset$
and $\hedge \setminus D \neq \emptyset$,
belongs to $E_1$.
\end{claim}

\begin{proof}
Assume that $\hedge \notin E_1$, that is, $\hedge$ is monochromatic in $\col_{opt}$.
As $\hedge \cap D \neq \emptyset$, $\col_{opt}$ needs to colour all elements of $\hedge$ white,
and $D \in E_{white}$.
However, since $E_0$ is a spanning forest of the hypergraph $(V(H),E_{white})$
and $p$ is a good assignment, then all the elements of $F$ are contained in the same connected
component of $L_p$, and consequently $\hedge \subseteq D$.
\end{proof}

From now on let us assume that $p$ is a good assignment.

Let us modify the assignment $p$ as follows; note that we modify also the graph $L_p$ along with $p$.
As long as possible perform one of the following
two operations, preferring the first one over the second one:
\begin{enumerate}
  \item If there exists an edge $\hedge \in E(H)$
  such that $\hedge \subseteq D$ for some connected component of $L_p$
  and $p(\hedge) > 0$, then set $p(\hedge) = 0$.
  \item If there exist vertices $v_1, v_2 \in D$ (potentially $v_1 = v_2$),
and hyperedges $\hedge_1,\hedge_2 \in E(H)$ (potentially $\hedge_1=\hedge_2)$ intersecting a connected component $D$ of $L_p$, 
    such that $\hedge_1 \setminus D \neq \emptyset$, $\hedge_2 \setminus D \neq \emptyset$, $p(\hedge_1)=i > 0$, $p(\hedge_2) = j > 0$,
  and $\col_{\hedge_1,i}(v_1)=\W$ and $\col_{\hedge_2,j}(v_2) = \B$,
  then set $p(\hedge_1)=0$.
\end{enumerate}
As in each round the number of edges of $E(H)$ assigned zeros is strictly increasing, the process
finishes in polynomial time.

\begin{claim}
After each step $p$ remains a good assignment.
\end{claim}

\begin{proof}
First assume that the first type of operation was performed.
By Claim~\ref{claim:hp:1} all the vertices
of $D$ are coloured with the same colour by $\col_{opt}$,
hence in particular all the vertices
of $\hedge$ are coloured with the same colour in $\col_{opt}$
and definitely $\hedge \not\in E_1$. Hence it is safe
to set $p(\hedge) = 0$.

Now assume that the second type of operation was performed.
We want to show that $\hedge_1$ is monochromatic in $\col_{opt}$.
Assume the contrary, i.e., $\hedge_1 \in E_1$.
Since $p$ was a good assignment (before the operation) we have $\col_{opt}|_{\hedge_1} = \col_{\hedge_1,p(\hedge_1)}$, which together with Claim~\ref{claim:hp:1}
implies that all the vertices of $D$ are coloured white by $\col_{opt}$.
However by Claim~\ref{claim:hp:2} this means that $\hedge_2 \in E_1$,
and as $p$ is a good assignment this means that $\col_{opt}$
colours all the vertices of $D$ black, a contradiction.
\end{proof}

We call a connected component $D$ of $L_p$
a \emph{black component} if there exists an edge $\hedge_2 \in E(H)$,
such that $p(\hedge_2) = i > 0$
and $\col_{\hedge_2,i}^{-1}(\B) \cap D \neq \emptyset$.
Otherwise we call $D$ a \emph{potentially white} component.
The next two claims show that these names are in fact meaningful.

\begin{claim}
\label{claim:hp:3}
For any black component $D$ of $L_p$, $\col_{opt}$ colours all vertices of $D$ black.
\end{claim}

\begin{proof}
Let $\hedge$ be an edge witnessing $D$ is a black component.
Note that $\hedge \not\subseteq D$, as otherwise the first operation would be applicable to $\hedge$.
By Claim~\ref{claim:hp:1}, $\col_{opt}$ colours all vertices of $D$ in the same colour.
If this colour is white, then, by Claim~\ref{claim:hp:2}, $\hedge \in E_1$, and, as $p$
is a good assignment, $\col_{\hedge,p(\hedge)} = \col_{opt}|_\hedge$.
However, this contradicts the assumption that $\col_{\hedge,p(\hedge)}$ colours some
vertex of $D$ black, and, consequently, $\col_{opt}$ colours $D$ black.
\end{proof}

\begin{claim}
\label{claim:hp:4}
Let $D$ be a potentially white component of $L_p$
and let $E_D \subseteq E(H)$ be the subset 
of edges with non-empty intersection with $D$.
Exactly one of the following conditions holds:
\begin{itemize}
\item $\col_{opt}$ colours all vertices of $D$ white, and
for each edge $\hedge \in E_D$ either
\begin{itemize}
\item $\hedge \not\subseteq D$, $p(\hedge) > 0$, $\hedge \in E_1$, $\col_{opt}|_\hedge = \col_{\hedge,p(\hedge)}$, or
\item $\hedge \subseteq D$, $p(\hedge)=0$ and $\col_{opt}$ colours all vertices of $\hedge$ white;
\end{itemize}
\item $\col_{opt}$ colours $D$ and each each edge of $E_D$ entirely black.
\end{itemize}
\end{claim}

\begin{proof}
Let $E' \subseteq E_D$ be the subset of edges of $E_D$
which are not fully contained in $D$.
By Claim~\ref{claim:hp:1}, $\col_{opt}$ colours $D$ monochromatically.

Assume first that $\col_{opt}$ colours all vertices of $D$ white, and
consider $\hedge \in E_D$. If $\hedge \subseteq D$ then $p(\hedge) = 0$
by the application of the first operation, and $\hedge \in E_{white}$.
If $\hedge \not\subseteq D$ then, by Claim~\ref{claim:hp:2}, $\hedge \in E_1$.
Since $p$ is good, $p(\hedge) > 0$ and $\col_{opt}|_\hedge = \col_{\hedge,p(\hedge)}$.

We are left with the case when $\col_{opt}$ colours all vertices of $D$ black.
Consider $\hedge \in E_D$.
If $p(\hedge) = 0$ then the assumption that $p$ is good
implies that $\col_{opt}$ colours $\hedge$ monochromatically; as $\hedge \cap D \neq \emptyset$
then $\hedge$ is coloured black by $\col_{opt}$.
If $p(\hedge) = i > 0$ then, since $D$ is potentially white,
$\col_{\hedge,i}(v) = \W \neq \col_{opt}(v)$
for any $v \in \hedge \cap D$. Consequently,
$\col_{opt}|_\hedge \neq \col_{\hedge,i}$ and, since $p$ is good, $\hedge \notin E_1$.
Therefore $\col_{opt}$ colours $\hedge$ monochromatically, and, since $\hedge \cap D \neq \emptyset$,
it colours $\hedge$ black.
\end{proof}

Let $E_{black}$ be the set of all edges of $E(H)$
contained in black components of $L_p$. The following claim states that we may consider sets $E_{black}$ and $E_D$ for different potentially white components $D$ independently.

\begin{claim}\label{claim:hp:5}
Every edge $\hedge\in E(H)$ belongs to exactly one of the sets: to $E_{black}$ or to one of the sets $E_D$ for potentially white connected components $D$ of $L_p$.
\end{claim}
\begin{proof}
Assume first that $\hedge\subseteq D$ for some connected component $D$ of $L_p$. Then either $D$ is black and $\hedge\in E_{black}$, or $D$ is potentially white and $\hedge\in E_D$.

Assume now that $\hedge$ is not entirely contained in any connected component of $L_p$. By the construction of $L_p$, we have that $p(\hedge)=i>0$ and $\hedge$ intersects exactly two different components $D_1,D_2$ of $L_p$, such that w.l.o.g. $\col_{\hedge,i}^{-1}(\W)=D_1\cap \hedge$ and $\col_{\hedge,i}^{-1}(\B)=D_2\cap \hedge$. To prove the claim it suffices to show that (a) $D_1$ is potentially white, and (b) $D_2$ is black, as then $\hedge$ will belong only to $E_{D_1}$ among the sets present in the statement of the claim. For (a), observe that otherwise the second operation would set $p(\hedge)=0$, and (b) follows directly from the definition of being black.
\end{proof}

Armed with Claims~\ref{claim:hp:3},~\ref{claim:hp:4} and~\ref{claim:hp:5}, we proceed to presenting the algorithm. First, we need to include the constraints imposed by the colouring $\col_0$.
To this end, for any hyperedge $\hedge$, $0 \leq \accum \leq b$ and
 any colouring $\col: \hedge \to \{\B,\W\}$
we set
$$\hat{f}_\hedge(\col,\accum) = \begin{cases}
\duzo & \textrm{ if } \exists_{v \in \hedge}\ \col(v) \neq \col_0(v) \\
f_\hedge(\col,\accum) & \textrm {otherwise.}
\end{cases}$$
That is, we set the cost of colouring $\hedge$ with $\col$
as $\duzo$ whenever $\col$ conflicts with $\col_0$ on some vertex.

Now we handle edges contained entirely in black components.
For an edge $\hedge \in E_{black}$
let $\hat{f}_{\hedge}^{black}:\{0,\ldots,b\} \to \{0,1,\ldots,k,\duzo\}$ be the function
$$\hat{f}_{\hedge}^{black}(\accum) = \hat{f}_{\hedge}(\{\B\}^\hedge,\accum)\,.$$

Let $t: \{0,\ldots,b\} \to \{0,1,\ldots,k,\duzo\}$ be a function
such that $t(0) = 0$ and $t(\accum)=\duzo$ for $\accum> 0$.
For each edge $\hedge \in E_{black}$ we update the function $t$
by setting $t:=t \oplus \hat{f}_{\hedge}^{black}$.
It remains to process all the edges $E(H) \setminus E_{black}$.

Consider all the white components $D$ of $L_p$ one by one.
Let $t_1,t_2: \{0,\ldots,b\} \to \{0,1,\ldots,k,\duzo\}$ be functions
such that $t_1(0) = t_2(0) = 0$ and $t_1(\accum)=t_2(\accum)=\duzo$ for $\accum> 0$.
We want to make $t_1$ represent the case when all the edges of $E_D$ are black,
while $t_2$ represent the other case of Claim~\ref{claim:hp:4}.
First, for each edge $\hedge \in E_D$ set $t_1:=t_1 \oplus \hat{f}_\hedge(\{\B\}^\hedge,\cdot)$.
Moreover, for each edge $\hedge \in E_D$ such that $p(\hedge)=0$
do $t_2:= t_2 \oplus \hat{f}_\hedge(\{\W\}^\hedge,\cdot)$,
while for each edge $\hedge \in E_D$ such that $p(\hedge) > 0$
do $t_2:= t_2 \oplus \hat{f}_\hedge(\col_{\hedge,p(\hedge)},\cdot)$.


Finally make the update 
$$t:=\min(t \oplus t_1, t \oplus t_2)\,.$$
At the end of the process the value $t(\accum)$ equals $w_\accum$
and the correctness of our algorithm follows from Claim~\ref{claim:hp:3}, Claim~\ref{claim:hp:4}, and Claim~\ref{claim:hp:5}.
\end{proof}

\subsection{Dynamic programming}
\label{ssec:bisection-dp}

In this section we show that by
constructing a tree decomposition
from Theorem~\ref{thm:decomposition-main}
and invoking the algorithm of Lemma~\ref{lem:hp}
one can solve the {\sc Minimum Bisection}
problem in $\Ohstar(2^{O(k^3)})$ time, proving Theorem~\ref{thm:bisectionFPT}
.

\begin{proof}[Proof of Theorem~\ref{thm:bisectionFPT}]
First note that, without loss of generality, we may focus on the following variant:
the input graph $G$ is required to be connected, and our goal is to partition $V(G)$ into parts $A$ and $B$ of prescribed size
minimizing $|E(A,B)|$. The algorithm for the classic \textsc{Minimum Bisection} problem follows
from a standard knapsack-type dynamic programming on connected components of the input graph.

As the input graph is connected, we may use Theorem~\ref{thm:decomposition-main}.
Let $(T,\beta)$ be a tree decomposition
constructed by the algorithm of Theorem~\ref{thm:decomposition-main}
in $\Ohstar(2^{O(k^2)})$ time.

As usual for tree decompositions, we will
use a dynamic programming approach.
For a node $t \in V(T)$ of the tree decomposition,
an integer $\accum$, $0 \le \accum \le n$,
and a colouring $\col_0 : \sigma(t) \to \{\B,\W\}$
satisfying 
\begin{align}
\label{eq:bisection:1}
\min(|\col_0^{-1}(\B)|, |\col_0^{-1}(\W)|) \le 3k\,,
\end{align}
we consider a variable $x_{t,\col_0,\accum}$.

The variable $x_{t,\col_0,\accum}$
equals the minimum cardinality of a set
$Z \subseteq E(G[\gamma(t)])$,
such that there exists a colouring $\col : \gamma(t) \to \{\B,\W\}$,
where $\col|_{\sigma(t)} = \col_0$,
no edge of $E(G[\gamma(t)]) \setminus Z$
is incident to two vertices of different colours in $\col$,
and the total number of white vertices equals $\accum$,
i.e. $|\col^{-1}(\W)| = \accum$.
Additionally if it is impossible to find such a colouring $\col$,
or the number of edges one needs to include in $Z$ is greater than $k$,
then we define $x_{t,\col_0,\accum} = \duzo$.
The restriction (\ref{eq:bisection:1}) of $\col_0$ 
will be used to optimize the running time.

As Theorem~\ref{thm:decomposition-main} upper bounds
the cardinality of $\sigma(t)$ by $2^{O(k)}$,
the total number of values $x_{t,\col_0,\accum}$
we want to compute is $\Oh(2^{O(k^2)} n^2)$.
Note that having all those values is enough
to solve the considered variant of the {\sc Minimum Bisection} problem
as the minimum possible size of the cut $E(A,B)$
equals $x_{r,\emptyset,a}$, where $r$
is the root of $(T,\beta)$, $\emptyset$
plays the role of the single colouring
of $\sigma(r) = \emptyset$ and $a$ is the prescribed size of one part of the partition we are looking for.
The value $x_{r,\emptyset,a}$ attains $\duzo$ if any feasible cut $E(A,B)$ is of size larger than $k$.
We will compute the values $x_{t,\cdot,\cdot}$
in a bottom-up manner, that is our computation
is performed for a node $t \in V(T)$ only
after all the values $x_{t',\cdot,\cdot}$
for $t' \prec t$ have been already computed.

Consider a fixed $t \in V(T)$
and a colouring $\col_0 : \sigma(t) \to \{\B,\W\}$
satisfying~(\ref{eq:bisection:1}).
In what follows we show how to find
all the values $x_{t,\col_0,\cdot}$
by solving a single proper instance of the \hpname
problem.
Create an auxiliary hypergraph $H$,
with a vertex set $V(H) = \beta(t)$ 
and the following set of edges; in the following we use Iverson notation, i.e., $[\varphi]$ is equal to $1$ if the condition $\varphi$ is true and $0$ otherwise.
\begin{enumerate}[(a)]
  \item\label{he:type1} For each vertex $v \in \beta(t)$
  add to $H$ a hyperedge $\hedge=\{v\}$,
  and define a function $f_\hedge : \{\B,\W\}^\hedge \times \{0,\ldots,n\} \to \{0,1,\ldots,k,\duzo\}$
$$f_\hedge(\col_F,\accum) = \begin{cases}
0 & \text{ if } \accum=[\col_F(v) = \W],\\
\duzo & \text{ otherwise.}  
\end{cases}$$
  We introduce those edges in order to keep track of the number of white vertices in $\beta(t)$.
  \item\label{he:type2}For each edge $uv \in E(G[\beta(t)])$
  add to $H$ a hyperedge $\hedge=\{u,v\}$,
  and define a function $f_\hedge : \{\B,\W\}^\hedge \times \{0,\ldots,n\} \to \{0,1,\ldots,k,\duzo\}$
$$f_\hedge(\col_F,\accum) = \begin{cases}
[\col_F(u) \neq \col_F(v)] & \text{ if } \accum=0,\\
\duzo & \text{ otherwise.}  
\end{cases}$$
  We introduce those edges in order to keep track of the number of
 edges with endpoints of different colours in $G[\beta(t)]$.
  \item\label{he:type3} For each $t' \in V(T)$ which 
  is a child of $t$ in the tree decomposition
  add to $H$ a hyperedge $\hedge=\sigma(t')$,
  and define a function $f_\hedge : \{\B,\W\}^\hedge \times \{0,\ldots,n\} \to \{0,1,\ldots,k,\duzo\}$
$$f_\hedge(\col_F,\accum) = 
\begin{cases}
\duzo & \text{ if } \min(|\col_F^{-1}(\B)|,|\col_F^{-1}(\W))| > 3k \text{ or } x_{t',\col_F,\accum+\accum_0} = \duzo\\
x_{t',\col_F,\accum+\accum_0} - x_0  & \text{ otherwise,}
\end{cases}
$$
where $\accum_0=|\col_F^{-1}(\W)|$, 
and $x_0=|\{uv \in E(G[\sigma(t')]) : \col_F(u) \neq \col_F(v)\}|$.
Less formally, we are shifting values $x_0$ and $\accum_0$ 
in order not to overcount white vertices of $\sigma(t')$
and edges of $G[\sigma(t')]$ having endpoints of different colours in $\col_F$,
as a vertex of $\sigma(t')$ might appear in several bags,
and similarly an edge of $G[\sigma(t')]$ may have both endpoints in several bags
being children of $t'$. 
\end{enumerate}
Note that each of the edges of $H$ is 
of size at most $\eta$ (by Theorem~\ref{thm:decomposition-main}),
hence $$I=(k,n,\eta,q,H,\col_0,(f_\hedge)_{\hedge \in E(H)})$$
is an instance of the \hpname problem for any $q$,
which we are about to define.

\begin{claim}
\label{claim:bisection:helpful}
Let $(w_\accum)_{0 \le \accum \le n}$ be the solution for the instance
$I$ of the \hpname problem.
Then for any $0 \le \accum \le n$ we have $x_{t,\col_0,\accum} = w_\accum$.

Moreover for any colouring $\col : \beta(t) \to \{\B,\W\}$
witnessing $w_\accum \le k$ there is an extension $\col' : \gamma(t) \to \{\B,\W\}$,
such that $\col'|_{\beta(t)} = \col$, and
the number of bichromatic edges of $G[\gamma(t)]$ with respect to $\col'$
equals $w_\accum$.
\end{claim}

\begin{proof}
Fix an arbitrary $0 \le \accum \le n$.
First, we show that $x_{t,\col_0,\accum} \ge w_\accum$.
Note that the inequality trivially holds for $x_{t,\col_0,\accum} = \duzo$,
hence let us assume $x_{t,\col_0,\accum} \le k$
and let $\col : \gamma(t) \to \{\B,\W\}$ be
a colouring such that 
\begin{itemize}
  \item $\col|_{\sigma(t)} = \col_0$,
  \item $|Z| \le k$, where $Z=\{uv \in E(G[\gamma(t)]) : \col(u) \neq \col(v)\}$,
  \item $|\col^{-1}(\W)| = \accum$.
\end{itemize}
Recall that Theorem~\ref{thm:decomposition-main} ensures that 
for any child $t'$ of $t$ in the tree decomposition
the adhesion $\sigma(t')$ is $(2k,k)$-unbreakable
in $G[\gamma(t)]$.
Therefore, 
\begin{align*}
\min(|\col^{-1}(\B) \cap \sigma(t')|,|\col^{-1}(\W) \cap \sigma(t')|) \le 3k\,,
\end{align*}
as otherwise $(X=N_{G[\gamma(t)]}[\col^{-1}(\B)],Y=\col^{-1}(\W))$
would be a separation of $G[\gamma(t)]$ of order at most $k$ with $|(X\setminus Y)\cap \sigma(t')|,|(Y\setminus X)\cap \sigma(t')|>2k$,
contradicting the fact that $\sigma(t')$ is $(2k,k)$-unbreakable in $G[\gamma(t)]$.
Consequently, the values $x_{t',\col|_{\sigma(t')},\cdot}$ are well-defined,
i.e., $\col|_{\sigma(t')}$ satisfies (\ref{eq:bisection:1}).
Furthermore, observe that
\begin{align}
\label{eq:bisection:2}
x_{t',\col|_{\sigma(t')},|\col^{-1}(\W) \cap \gamma(t')|} \le |Z \cap E(G[\gamma(t')])|,
\end{align}
which is witnessed by the colouring $\col|_{\gamma(t')}$.
For $\hedge \in E(H)$ define $a_\hedge = |\col^{-1}(\W) \cap (\gamma(t') \setminus \sigma(t'))|$.
Next, we verify that $\col|_{\beta(t)}$ and $(a_\hedge)_{\hedge \in E(H)}$
certify that $x_{t,\col_0,\accum} \ge w_\accum$.
We split the contributions of edges of $E(H)$
to the sum $\sum_{\hedge \in E(H)} f_\hedge(\col|_F, a_F)$
into three summands, according to the types of edges of $E(H)$.
The edges of type (\ref{he:type1}) do not contribute to the 
sum at all.
The edges of type (\ref{he:type2}) contribute exactly
$|Z \cap E(G[\beta(t)])|$,
while the edges of type (\ref{he:type3}) contribute exactly 
$\sum_{t'} x_{t',\col|_{\sigma(t')},a_F + |\col^{-1}(\W) \cap \sigma(t')|} - |Z \cap E(G[\sigma(t')])| \le \sum_{t'} |Z \cap (E(G[\gamma(t')] \setminus E(G[\sigma(t')])))|$, where the sum is over all children $t'$ of $t$ and the inequality follows from (\ref{eq:bisection:2}).
This means that each edge of $Z$ is counted exactly once,
so the total contribution is at most $|Z| = x_{t,\col_0,\accum}$.

In the other direction, we want to show $x_{t,\col_0,\accum} \le w_\accum$.
As in the previous case, for $w_\accum = \infty$ the inequality trivially holds.
Hence, we assume $w_\accum \le k$.
Let $\col : \beta(t) \to \{\B,\W\}$ be a colouring and $\sum_{\hedge \in E(H)} a_\hedge$ be
a partition of $\accum$ witnessing the value of $w_\accum$, i.e., satisfying
\begin{itemize}
  \item $\col|_{\sigma(t)} = \col_0$,
  \item $\sum_{\hedge \in E(H)} a_\hedge = \accum$,
  \item $w_\accum = \sum_{\hedge \in E(H)} f_\hedge(\col|_\hedge,a_\hedge)$.
\end{itemize}
Our goal is to extend the colouring $\col$ on $\gamma(t) \setminus \beta(t)$, so that the total number of white vertices equals $\accum$
and the number of bichromatic edges equals $w_\accum$.
Initially set $\col' = \col$ and
consider children $t'$ of $t$ in the tree decomposition one by one.
Let $\hedge = \sigma(t') \in E(H)$
be the type~(\ref{he:type3}) edge of $H$.
Since $w_\accum \leq k$, we have $f_\hedge(\col|_\hedge, \accum_\hedge) \leq k$,
and by the definition of $f_F$ 
$$f_{\hedge}(\col|_{\hedge}, \accum_{\hedge}) = x_{t',\col|_{\hedge},\accum_{\hedge}-\accum_0}-x_0\,,$$
where $\accum_0 = |\col^{-1}(\W) \cap \sigma(t')|$,
and $x_0 = |\{uv \in E(G[\sigma(t')]) : \col(u) \neq \col(v)\}|$.
Let $\col_\hedge : \gamma(t') \to \{\B,\W\}$ 
be the colouring witnessing the value 
$x_{t',\col|_\hedge,\accum_\hedge-\accum_0}$.
Note that $\col_\hedge$ is consistent with $\col$ on $\hedge=\sigma(t')$,
so we can update $\col'$ by setting $\col'=\col' \cup \col_\hedge$.

Observe that the edges of $E(H)$ of type~(\ref{he:type1})
together with shifting by $\accum_0$ 
ensure that $\col'$ colours exactly $\accum$ vertices white.
Finally, the edges of $E(H)$ of type~(\ref{he:type2})
together with shifting by $x_0$ 
ensure that $\col'$ has exactly $w_\accum$ bichromatic 
edges, which shows $x_{t,\col_0,\accum} \le w_\accum$.
As $\col'|_{\beta(t)} = \col$ the last part of the claim follows as well.
\end{proof}

The previous claim shows that solving the \hpname instance $I$
is enough to find the values $x_{t,\col_0,\cdot}$,
however in the previous section we have only shown
how to solve proper instances of \hpname.
Therefore, we show that there is a small enough 
value of $q$, such that $I$ becomes a proper instance.

\begin{claim}
\label{claim:bisection:proper}
There is $q = 2^{O(k)}$
such that $I$ is a proper instance of the \hpname problem.
\end{claim}

\begin{proof}
For the hyperedges $\hedge \in E(H)$ of size at most
two the local unbreakability property is trivially satisfied,
while for all the other hyperedges $\hedge=\sigma(t')$
local unbreakability follows directly 
from the definition of $f_\hedge$.

By Theorem~\ref{thm:decomposition-main} 
each $G[\gamma(t')] \setminus \sigma(t')$ is connected and $N(\gamma(t')\setminus \sigma(t'))=\sigma(t')$,
which means that the graph $G[\gamma(t')] \setminus E(G[\sigma(t')])$ is connected, 
and consequently $x_{t',\col_F,\cdot} > 0$ 
for any colouring $\col_F$ which uses both colours
(as we need to remove at least one edge).
This proves the connectivity property.

Recall, that by Theorem~\ref{thm:decomposition-main}
the set $\beta(t)$ is $(\tau',k)$-unbreakable in $G[\gamma(t)]$
for some $\tau'=2^{O(k)}$.
Let $q = \tau' + k$
and let $(w_\accum)_{0 \le \accum \le n}$ be a solution
for the instance $I$ of the \hpname problem.
Consider an arbitrary $0 \le \accum \le n$ such that $w_\accum \le k$.
We want to show, that there exists a witnessing colouring $\col : \beta(t) \to \{\B,\W\}$ certifying the global unbreakability.
In fact we will show that any colouring $\col : \beta(t) \to \{\B,\W\}$
witnessing $w_\accum$ satisfies
$$\min(|\col^{-1}(\B)|,|\col^{-1}(\W)|) \le q = \tau'+k\,.$$
By Claim~\ref{claim:bisection:helpful}\footnote{Note that to use Claim~\ref{claim:bisection:helpful}
we do not require that $I$ is proper.}
there is an extension $\col'$ of $\col$, having $w_\accum \le k$
bichromatic edges of $G[\gamma(t)]$.
Note that $(X=N_{G[\gamma(t)]}[\col'^{-1}(\B)],Y=\col'^{-1}(\B))$ is a separation of $G[\gamma(t)]$ of order at most $k$,
hence by $(\tau',k)$-unbreakability of $\beta(t)$
we have 
$$\min(|(X\setminus Y)\cap \beta(t)|,|(Y\setminus X)\cap \beta(t)|) \le \tau'\,.$$
However $|\col^{-1}(\B)| \le |(X \setminus Y) \cap \beta(t)| + k$
and $|\col^{-1}(\W)| \le |(Y \setminus X) \cap \beta(t)| + k$,
which implies 
$$\min(|\col^{-1}(\B)|,|\col^{-1}(\W)|) \le k+\tau' = q\,,$$
proving the global unbreakability property.
\end{proof}

By Claim~\ref{claim:bisection:proper} we can 
use Lemma~\ref{lem:hp} and in $\Ohstar(q^{\Oh(k)} \cdot \eta^{\Oh(k^2)})=\Ohstar(2^{\Oh(k^3)})$ time 
compute the values $w_{\accum}$ for each $0 \le \accum \le n$.
At the same time Claim~\ref{claim:bisection:helpful} shows
that $x_{t,\col_0,\accum}=w_{\accum}$ for each $0 \le \accum \le n$.
Since the number of nodes of $V(T)$ is at most $|V(G)|$
and the number of colourings obeying ($\ref{eq:bisection:1})$
is $2^{\Oh(k^2)}$ the whole
dynamic programming routine takes $\Ohstar(2^{\Oh(k^3)})$ time.
Consequently Theorem~\ref{thm:bisectionFPT} follows.
\end{proof}

\subsection{Dependency on the size of $G$ in the running time}\label{ssec:bisection-ptime}

Here we argue about the factors polynomial in the size of $G$
in the running time of the algorithm.

We first note that we may assume that $m = |E(G)| = \Oh(kn)$, by applying the sparsification
technique of Nagamochi and Ibaraki~\cite{nagamochi-ibaraki}.

\begin{lemma}[\cite{nagamochi-ibaraki}]
\label{lem:japanese}
Given an undirected graph $G$ and an integer $k$, in $O(k(|V(G)|+|E(G)|))$ time we can obtain 
a set of edges $E_0 \subseteq E(G)$ of size at most $(k+1)(|V(G)|-1)$,
such that for any edge $uv \in E(G)\setminus E_0$ in the graph
$(V(G),E_0)$ there are at least $k+1$ edge-disjoint paths between $u$ and $v$.
\end{lemma}

\begin{proof}
The algorithm performs exactly $k+1$ iterations.
In each iteration it finds a spanning forest $F$ of the graph $G$, adds
all the edges of $F$ to $E_0$ and removes all the edges of $F$ from the graph $G$.

Observe that for any edge $uv$ remaining in the graph $G$, the vertices $u$ and $v$
are in the same connected components in each of the forests found.
Hence in each of those forests we can find a path between $u$ and $v$;
thus, we obtain $k+1$ edge-disjoint paths between $u$ and $v$.
\end{proof}

The above lemma allows us to sparsify the graph, so that it contains $\Oh(kn)$ edges,
and any edge cut of size at most $k$ remains in the graph, while any edge cut
with at least $k+1$ edges after sparsification still has at least $k+1$ edges.
Therefore applying Lemma~\ref{lem:japanese} gives us an equivalent instance $(V(G), E_0)$
and consequently the construction of the decomposition takes $\Oh(2^{\Oh(k^2)} n^3)$ time.

There are at most $n$ bags of the decomposition, which adds a $\Oh(n)$ factor to the running
time. In each bag $t$, we consider $\Oh(\eta^{\Oh(k)})$ colourings of the adhesion $\sigma(t)$;
hence, there are $\Oh(\eta^{\Oh(k)} n)$ calls to the procedure
solving \hpname.

In each call, we have $V(H) = \beta(t)$ and $|E(H)| = \Oh(n+m) = \Oh(kn)$, as we have a hyperedge for
each vertex and edge of $\beta(t)$ as well as an edge for each child of $t$ in the decomposition.
As discussed in Section~\ref{ssec:bisection-problem}, each function $f$
can be represented by giving $\Oh(\eta^{\Oh(k)} n)$ values different than $\duzo$.

Note that we do not need to perform the entire algorithm for \hpname for each value of $\accum$
independently. Instead, we may perform it only once, and return $w_\accum$ to be the minimum
$t(\accum)$ among all branches of the algorithm.

By Lemma~\ref{lem:splitters} and the construction of perfect families of \cite{naor-schulman-srinivasan-derandom}, there are $\Oh(2^{\Oh(k^3)} \log^2 n)$ choices of the assignment $p$,
and they can be enumerated in $\Oh(2^{\Oh(k^3)} n \log^2n)$ time.

For each assignment $p$, we need to perform the two operations exhaustively.
To speed them up, instead of maintaining the entire graph $L_p$, we keep only its connected
components: each vertex of $H$ knows its connected component, and the connected component
knows its size and its vertices.
In this manner, by enumerating the smaller component, we can merge 
two connected components in amortized $\Oh(\log n)$ time,
as each vertex changes the connected components it belongs to $\Oh(\log n)$ times.
Consequently we may initiate the graph $L_p$ in $\Oh(\eta k n + n \log n)$ time,
as we have to iterate over $\Oh(kn)$ edges of size $\eta$ each,
and the total time needed to merge connected components is $\Oh(n \log n)$.

To apply the operations, we maintain the following auxiliary information.
Each hyperedge $\hedge$ stores a set of the connected components of $L_p$ it intersects
(note that this set is of size at most $2$).
Once this set changes its cardinality from $2$ to $1$,
the first operation starts to be applicable on $\hedge$.
As each vertex changes its connected component $\Oh(\log n)$ times, each list is updated
at most $\Oh(\log n)$ times,
which gives $\Oh(kn \log n)$ time in total.

For the second operation, we need to maintain, for each connected component $D$ of $L_p$,
a list $T(D,\B)$ of hyperedges $\hedge$ such that 
$\hedge \cap D \neq \emptyset$, $\hedge \setminus D \neq \emptyset$,
$p(\hedge) > 0$ and $\col_{\hedge,p(\hedge)}(v) = \B$ for some $v \in \hedge \cap D$;
analogously we define a list $T(D,\W)$.
Once both lists are non-empty, the second operation is applicable.
As each hyperedge is of size at most $\eta$, all lists can be recomputed
in $\Oh(\eta kn)$ time, whenever the set of the connected components of the graph $L_p$ changes:
each hyperedge $\hedge$ inserts itself into at most $\eta$ lists.

We infer that the operations can be exhaustively applied in $\Oh(2^{\Oh(k)} n^2)$ time
for a fixed assignment $p$.
Also, including the constraints imposed by the colouring $\col_0$,
i.e., obtaining the functions $\hat{f}_\hedge(\col,\accum)$,
can be done in $\Oh(2^{\Oh(k^2)} n^2)$ time.

We now move to the analysis of the final knapsack-type dynamic programming routine.
We first show that the $\oplus$ operation can be performed using $\Oh(k^2)$ applications of the Fast Fourier Transform, taking total time $\Oh(k^2 b \log b)$, instead of the naive $\Oh(b^2)$ time bound.
Consider two functions $t_1,t_2 \in \{0, \ldots, b\} \to \{0,\ldots,k,\infty\}$.
For $i \in \{1,2\}$ and $0 \le j \le k$ by $p_{i,j}$ we 
define the polynomial
$$p_{i,j}(x) = \sum_{0 \le \accum \le b} [t_i(\accum) = j]x^{\accum}\,.$$
Note that if $(t_1 \oplus t_2) (\accum) \neq \infty$, 
then $(t_1 \oplus t_2) (\accum)$
is equal to the smallest $j$, such that 
for some partition $j = j_1 + j_2$ 
the coefficient in front of the monomial $x^{\accum}$
in the polynomial $p_{1,j_1} \cdot p_{2,j_2}$ is non-zero.
Therefore we can compute $t_1 \oplus t_2$ in $\Oh(k^2 b \log b$) time.
There are $\Oh(1)$ such operations per each edge of $E(H)$. Consequently,
the final dynamic programming algorithm takes $\Oh(2^{O(k)} n^2\log n)$ time.

We conclude that the total running time is
$\Oh(2^{\Oh(k^3)} n^3 \log^3n)$, as promised in Theorem~\ref{thm:bisectionFPT}.

%% file: weights.tex
In this section we sketch how using our approach one can solve the following weighted variant of the \textsc{Minimum Bisection} problem:

\begin{theorem}\label{thm:bisectionFPT-weighted}
Given a graph $G$ with edge weights $w:E(G) \to \R$ and an integer $k$,
one can in $\Ohstar(2^{k^3})$ time find a partition of $V(G)$ into sets $A$ and $B$
minimizing $\sum_{e \in E(A,B)} w(e)$ subject to $||A|-|B|| \leq 1$ and $|E(A,B)| \leq k$,
or state that such a partition does not exist.
\end{theorem}
\begin{proof}
Essentially, we follow the same approach as in the previous section, except that in all dynamic programming tables
we need to add an additional dimension to control the size of the constructed cut $E(A,B)$, and store the weight
of the cut as the value of the entry in the DP table.

In some more details, for a fixed bag $t$, a colouring $\col_0:\sigma(t) \to \{\B,\W\}$
satisfying~\eqref{eq:bisection:1}, and integers $0 \leq \accum \leq n$ and $0 \leq \xi \leq k$
we consider a variable $x_{t,\col_0,\accum,\xi} \in \R \cup \{+\infty\}$
that equals the minimum possible value of $\sum_{e \in E(\col^{-1}(\B),\col^{-1}(\W))} w(e)$
among colourings $\col:\gamma(t) \to \{\B,\W\}$ satisfying:
\begin{itemize}
\item $\col|_{\sigma(t)} = \col_0$,
\item $|\col^{-1}(\W)| = \accum$, and
\item $|E(\col^{-1}(\B), \col^{-1}(\W))| = \xi$.
\end{itemize}
The value $+\infty$ is attained if no such colouring exists.

Analogously, we modify the \hpname problem to match the aforementioned definition
of the values $x_{t,\col_0,\accum,\xi}$. That is, it takes as an input functions
$f_\hedge: \{\B,\W\}^\hedge \times \{0,\ldots,b\} \times \{0,\ldots,k\} \to \R \cup \{+\infty\}$,
where we require value $+\infty$ for any colouring that violates the local unbreakability constraint.
For each $0 \leq \accum \leq b$ and $0 \leq \xi \leq k$
we seek for a value $w_{\accum,\xi} \in \R \cup \{+\infty\}$
defined as a minimum, among all colourings
$\col$ extending $\col_0$, and all possible sequences $(a_\hedge)_{\hedge \in E(H)}$
and $(b_\hedge)_{\hedge \in E(H)}$ such that $\sum_\hedge a_\hedge = \accum$
and $\sum_\hedge b_\hedge = \xi$, of
$$\sum_{\hedge \in E(H)} f_\hedge(\col|_\hedge, a_\hedge, b_\hedge).$$
The knapsack-type dynamic programming of Section~\ref{ssec:bisection-problem}
is adjusted in a natural way, and
the remaining reasoning of Section~\ref{ssec:bisection-problem}
remains unaffected by the weights.
Consequently, the adjusted \hpname problem can be solved in 
$\Ohstar(q^{\Oh(k)} \cdot d^{O(k^2)})$ time.

It is straightforward to check that the adjusted \hpname problem
corresponds again to the task of handling one bag in the tree decomposition of the input graph.
To finish the proof note that the value we are looking for equals
$\min_{0 \leq \xi \leq k} x_{r,\emptyset,\lfloor n/2 \rfloor, \xi}$,
where $r$ is the root of the tree decomposition.
\end{proof}

%% file: alpha-sep.tex
In this section we argue that the algorithm of Section~\ref{sec:bisection}
can be extended to show the following:

\begin{theorem}\label{thm:alpha-sep}
Given an $n$-vertex graph $G$, a real $\alpha \in (0,1)$ and an integer $k$, one can
in $\Oh(2^{\Oh(k^3)} n^{\Oh(1/\alpha)})$ time decide if there exists
 a set $X$ of at most $k$ edges of $G$
such that each connected component of $G \setminus X$ has at most $\alpha n$ vertices.
\end{theorem}

To prove Theorem~\ref{thm:alpha-sep}, we need
the following lemma that can be seen as a generalization of Lemma~7.3 of~\cite{abs-1304-6321}.

\begin{lemma}\label{lem:grouping}
Let $\alpha\in (0,1)$ be a real constant and let $a_1,a_2,\ldots,a_n\in [0,\alpha]$ be reals such that $\sum_{\ell=1}^n a_\ell=1$. Then one can partition numbers $a_1,a_2,\ldots,a_n$ into $2\left\lceil\frac{1}{\alpha}\right\rceil-1$ groups (possibly empty), such that the sum of numbers in each group is at most $\alpha$.
\end{lemma}
\begin{proof}
Let $q=\lceil\frac{1}{\alpha}\rceil$. For $\ell=0,1,\ldots,n$, let $b_\ell=\sum_{i=1}^\ell a_i$. For $j=1,2,\ldots,q-1$, let $i_j$ be the unique index such that $b_{i_j-1}\leq j\cdot \alpha$ and $b_{i_j}>j\cdot \alpha$. Let us denote also $i_0=0$ and $i_q=n+1$; then also $b_{i_0}\geq 0\cdot \alpha$ and $b_{i_q-1}\leq q\cdot \alpha$. Define the following groups:
\begin{eqnarray*}
\{ & \{a_1,a_2,\ldots,a_{i_1-1}\}, & \\
& \{a_{i_1}\}, & \\
& \{a_{i_1+1},a_{i_1+2},\ldots,a_{i_2-1}\}, & \\
& \{a_{i_2}\}, & \\
& \ldots & \\
& \{a_{i_{q-2}+1},a_{i_{q-2}+2},\ldots,a_{i_{q-1}-1}\}, & \\
& \{a_{i_{q-1}}\}, & \\
& \{a_{i_{q-1}+1},a_{i_{q-1}+2},\ldots,a_n\} & \}.
\end{eqnarray*}
For every group of form $\{a_{i_j+1},a_{i_j+2},\ldots,a_{i_{j+1}-1}\}$ we have that
$$\sum_{\ell=i_j+1}^{i_{j+1}-1} a_\ell=b_{i_{j+1}-1}-b_{i_j}\leq (j+1)\alpha-j\alpha=\alpha.$$
On the other hand, for every group of form $\{a_{i_j}\}$ we have that $a_{i_j}\leq \alpha$ by the assumption that $a_{i_j}\in [0,\alpha]$. Hence, the formed groups satisfy the required properties.
\end{proof}

The following corollary is immediately implied by Lemma~\ref{lem:grouping}.

\begin{corollary}\label{cor:grouping}
Let $\alpha\in (0,1)$ be a real constant and let $H$ be a graph on $n$ vertices. Then the following conditions are equivalent:
\begin{enumerate}[(a)]
\item Each connected component of $H$ has at most $\alpha n$ vertices.
\item There exists a partition of $V(H)$ into $\zeta$ possibly empty sets $A_1,A_2,\ldots,A_\zeta$, where $\zeta=2\left\lceil\frac{1}{\alpha}\right\rceil-1$, such that $|A_i|\leq \alpha n$ for each $i=1,2,\ldots,\zeta$ and no edge of $H$ connects two vertices from different parts.
\end{enumerate}
\end{corollary}

Equipped with Corollary~\ref{cor:grouping}, we may now describe how to modify the algorithm
of Section~\ref{sec:bisection} to prove Theorem~\ref{thm:alpha-sep}. Most of the modifications are straightforward, hence we just sketch the consecutive steps.

\begin{proof}[Proof of Theorem~\ref{thm:alpha-sep}]
By Corollary~\ref{cor:grouping}, we may equivalently seek for a colouring of
$V(G)$ into $\zeta = 2\left\lceil \frac{1}{\alpha}\right\rceil-1$ colours, such that at most
$k$ edges connect vertices of different colours.
Essentially, we now proceed as in Section~\ref{sec:bisection}, but, instead of colouring vertices
into black and white, we use $\zeta$ colours, and we keep track of the number of vertices
coloured in each colour.

In some more details, for a fixed bag $t$, a colouring $\col_0: \sigma(t) \to \{1,\ldots,q\}$
satisfying
\begin{equation}
\label{eq:col0:alpha}
\exists_{1 \leq c \leq \zeta}\ |\col_0^{-1}(c)| \geq |\sigma(t)|-3k
\end{equation}
and a function $\accum:\{1,\ldots,\zeta\} \to \{0,\ldots,n\}$
we consider a variable $x_{t,\col_0,\accum} \in \{0,1,\ldots,k,\duzo\}$
that equals the minimum possible number of edges with endpoints coloured by different colours
by a colouring $\col$, among all colourings $\col:\gamma(t) \to \{1,2,\ldots,\zeta\}$ satisfying:
\begin{itemize}
\item $\col|_{\sigma(t)} = \col_0$, and
\item for each $1 \leq c \leq \zeta$, $|\col^{-1}(c)| = \accum(c)$.
\end{itemize}
The value $\duzo$ is attained if any such colouring yields more than $k$ edges with endpoints
of different colours.

Recall that, for any bag $t$, the adhesion $\sigma(t)$ is $(2k,k)$-unbreakable.
Similarly as in Claim~\ref{claim:bisection:helpful}, we infer that
if in a colouring $\col: \gamma(t) \to \{1,\ldots,\zeta\}$ at most $k$ edges have endpoints
painted in different colours, it needs to colour all but at most $3k$
vertices of $\sigma(t)$ with a single colour.
This motivates condition \eqref{eq:col0:alpha}.
Note that this requirement is only needed to obtain $2^{\mathrm{poly}(k)}$ dependency on $k$,
and, if it is omitted, the dependency will become doubly-exponential.

We now modify the \hpname problem to match the aforementioned definition of the 
values $x_{t,\col_0,\accum}$. That is, the problem takes as an input functions
$f_\hedge: \{1,\ldots,\zeta\}^\hedge \times \{0,\ldots,b\}^\zeta \to \{0,1,\ldots,k,\duzo\}$.
For each $\accum: \{1,\ldots,\zeta\} \to \{0,\ldots,b\}$ we seek for a value
$w_\accum \in \{0,1,\ldots,k,\duzo\}$ defined as a minimum, among all colourings
$\col:V(H) \to \{1,\ldots,\zeta\}$ extending $\col_0$, and all possible sequences $(a_\hedge^c)_{\hedge \in E(H), 1 \leq c \leq \zeta}$ such that $\sum_\hedge a_\hedge^c = \accum(c)$ for each $1 \leq c \leq \zeta$, of
$$\sum_{\hedge \in E(H)} f_\hedge(\col|_\hedge,(a_\hedge^c)_{1 \leq c \leq \zeta}).$$
The value of $\duzo$ is attained whenever the sum exceeds $k$.

In the local unbreakability constraint
we require that a value different than $\duzo$ can be attained only if
all but at most $3k$ elements of $\hedge$ are coloured in a single colour.
This corresponds to the previously discussed condition \eqref{eq:col0:alpha}
on valid colourings $\col_0$ of an adhesion $\sigma(t)$.
The connectivity requirement states that $f_\hedge(\col,\alpha)$ is non-zero
whenever $\col$ uses at least two colours: the corresponding colouring of the subgraph $\gamma(t)$
(as in the proof of Claim~\ref{claim:bisection:helpful})
needs to colour the endpoints of at least one edge with different colours, as $\gamma(t)$
is connected.
The global unbreakability constraint requires that whenever $w_\accum < \duzo$,
there is a witnessing colouring $\col$ that colours all but at most $\tau'+k$
vertices with a single colour.
This follows from the fact that, in our decomposition,
$\beta(t)$ is $(\tau',k)$-unbreakable, so any colouring of $\gamma(t)$ that colours
endpoints of at most $k$ edges with different colours needs to paint all but at most
$\tau'+k$ vertices of $\beta(t)$ with the same colour.

The core spirit of the reasoning of Section~\ref{ssec:bisection-problem}
remains in fact unaffected by this change.
However, for sake of clarity, we now describe the changes in more details.
We apply colour-coding to paint
the hyperedges with assignment $p: E(H) \to \{0,\ldots,\ell\}$, where $p(\hedge) = 0$
means ``definitely monochromatic'' and $p(\hedge) = i > 0$ means ``monochromatic or coloured
according to the colouring $\col_{\hedge,i}:\hedge \to \{1,2,\ldots,\zeta\}$''. The colourings
$\col_{\hedge,i}$ are required to comply with the (new) unbreakability constraint, thus there
are $\eta^{\Oh(k)} \zeta^{\Oh(k)}$ such colourings.
We guess a colour --- call it \emph{black} --- that will be the dominant colour in $\beta(t)$.
We require that for all hyperedges that are not monochromatic in the solution $\col_{opt}$
we have $p(\hedge) = i > 0$ and $\col_{\hedge,i} = \col_{opt}|_\hedge$
(i.e., we have guessed the correct
colouring of $\hedge$), and the ``definitely monochromatic'' hyperedges span a hyperforest of the graph
$(V(H), E_{not\ black})$, where $E_{not\ black}$ consists of all not-black monochromatic
hyperedges in the colouring $\col_{opt}$.

For a hyperedge $\hedge$, we insert to $L_p$ all edges of $\hedge \times \hedge$ if $p(\hedge) = 0$
and all edges between the vertices of the same colour of $\col_{\hedge,i}$ if $p(\hedge) = i > 0$.
It is straightforward to verify that, if $p$ is guessed correctly, then
the following holds:
\begin{enumerate}
\item all connected components of $L_p$ are painted monochromatically in $\col_{opt}$ (cf. Claim~\ref{claim:hp:1});
\item if a connected component $D$ of $L_p$ is not painted black in $\col_{opt}$, then all hyperedges
$\hedge$ such that $\hedge \cap D \neq \emptyset$ and $\hedge \setminus D \neq \emptyset$
are not coloured monochromatically by $\col_{opt}$ and, consequently, their colourings
$\col_{\hedge,p(\hedge)}$ conforms with $\col_{opt}$ (cf. Claim~\ref{claim:hp:2}).
\end{enumerate}
We may now classify any connected component $D$ of $L_p$ as ``definitely black'' (there exists
a hyperedge $\hedge$ with $p(\hedge) > 0$ such that $\col_{\hedge,p(\hedge)}$ colours at least
one vertex of $D$ black) and ``potentially not black'' (otherwise). By the aforementioned discussion
and a reasoning analogous to Claim~\ref{claim:hp:3}, a definitely black component is painted
black by $\col_{opt}$.

The clean-up operations on $p$ are defined as follows:
\begin{enumerate}
\item for any hyperedge $\hedge$ completely contained in some component $D$ of $L_p$, $\hedge$ is monochromatic in $\col_{opt}$, so set $p(\hedge) = 0$;
\item for any hyperedge $\hedge$ with $p(\hedge) > 0$, if $\col_{\hedge,p(\hedge)}$ colours
some vertex of a definitely black connected component of $L_p$ with a colour different than black, set
$p(\hedge) = 0$ as the guess on $\col_{\hedge,p(\hedge)}$ is clearly incorrect.
\end{enumerate}
After the operations are performed exhaustively, it is straightforward to verify that
an analogue of Claim~\ref{claim:hp:4} holds, stating that for any potentially not black component
$D$, either $\col_{opt}$ colours $D$ with not-black colour, being consistent with all hyperedges
$\hedge$ with $p(\hedge) > 0$ that intersect $D$, or $\col_{opt}$ colours $D$ and all intersecting
hyperedges black.

However, it is no longer true that the components of $D$ may be considered independently 
in the final knapsack-type dynamic programming.
Indeed, if there exists a hyperedge $\hedge$ that intersects two potentially not black components
$D_1$ and $D_2$ (and, hence, $p(\hedge) > 0$), then $\col_{opt}$ paints $D_1$ black if and only
if it paints $D_2$ black as well.
Consequently, we need to adjust the knapsack-type DP in the following way.
A black component is painted black, so we can proceed with them as previously.
Two potentially not black components $D_1$ and $D_2$ are \emph{entangled} if there exists
a hyperedge $\hedge$ intersecting both of them.
Now, observe that all components of
each connected component of the entanglement relation make a joint decision
on whether they are painted black or not, and these decisions are independent
between each other: the decision in one connected
component of the entanglement relation does not influence the decision
in another one.
This allows us to adjust he knapsack-type dynamic programming 
of Section~\ref{ssec:bisection-problem}, 
considering in a single step
all hyperedges intersecting all connected components of $L_p$ contained
in a single connected component of the entanglement relation.

It is straightforward to check that the adjusted \hpname problem
corresponds again to the task of handling one bag in the tree decomposition of the input graph.
To finish the proof note that the minimum size of the cut we are looking for equals
$$\min_{\accum: \{1,\ldots,q\} \to \{0,\ldots,\lfloor \alpha n \rfloor\}} x_{r,\emptyset,\accum},$$
where $r$ is the root of the tree decomposition, as long as the cut has size at most $k$.
\end{proof}

%% file: conclusion.tex
In this paper we have settled the parameterized complexity of {\sc Minimum Bisection}. Our algorithm also works in the more general setting when the edges are weighted, when the vertex set is to be partitioned into a constant number of parts rather than only two, and when the cardinality of each of the parts is given as input. 

The core component of our algorithm is a new decomposition theorem for general graphs. Intuitively, we show that it is possible to partition any graph in a tree-like manner using small separators so that each of the resulting pieces cannot be broken any further. This uncovered structure is very natural in the context of cut-problems, and we strongly believe that our decomposition theorem will find many further algorithmic applications.

Having settled the parameterized complexity of  {\sc Minimum Bisection} it is natural to ask whether the problem also admits a {\em polynomial kernel}, i.e. a polynomial-time preprocessing algorithm that would reduce the size of the input graph to some polynomial of the budget $k$. This question, however, has been already resolved by van Bevern et al.~\cite{vanBevern13}, who showed that {\sc Minimum Bisection} does not admit a polynomial kernel unless $\text{coNP} \subseteq \text{NP}/\text{poly}$. We conclude with a few intriguing open questions.


\begin{enumerate}[(a)]
\item Can the running time of our algorithm be improved? In particular, does there exist an algorithm for {\sc Minimum Bisection} with running time $2^{\Oh(k)} n^{\Oh(1)}$, that is with linear dependence on the parameter in the exponent? 
\item The running time dependence of our algorithm on the input size is roughly cubic. Is it possible to obtain a fixed-parameter tractable algorithm with quadratic, or even nearly-linear running time dependence on input size? Note that the best known algorithm for graphs of bounded treewidth has quadratic dependence on the input size~\cite{JansenKLS05}.
\item Are the parameters in the decomposition theorem tight? For example, is it possible to lower the adhesion size from $2^{\Oh(k)}$ to polynomial in $k$? Similarly, can one make the bags $(k^{\Oh(1)}, k)$-unbreakable rather than $(2^{\Oh(k)},k)$-unbreakable? Is it possible to achieve both simultaneously? We remark that if the latter question has a positive answer, this would improve the parameter dependence in the running time of our algorithm for {\sc Minimum Bisection} to $k^{\Oh(k)}$.
\item Is it possible to compute our decomposition faster, say in $2^{\Oh(k \log k)}n^{\Oh(1)}$ or even in $2^{\Oh(k)} n^{\Oh(1)}$ time? Currently the main bottleneck is the very simple Lemma~\ref{lem:breaking}, which we are unable to speed up.
\end{enumerate}